\documentclass[runningheads]{llncs}

\usepackage{tabularx}
\usepackage{multicol}
\usepackage{tabularx}
\usepackage{graphicx}
\usepackage{hyperref}
\usepackage{algorithmic}
\usepackage{algorithm}
\usepackage{array}
\usepackage{caption}
\usepackage{textcomp}
\usepackage{stfloats}
\usepackage{url}
\usepackage{verbatim}
\usepackage{graphicx}
\usepackage[square,numbers,sort&compress]{natbib}

\usepackage{graphicx}

\usepackage{float}
\usepackage{amsmath}
\usepackage{hyperref} 
\usepackage{amsmath}
\allowdisplaybreaks[4]
\usepackage{mathrsfs}
\usepackage{amssymb} 
\usepackage{bbm,dsfont}   
\usepackage{enumerate} 

\usepackage[table]{xcolor}
\usepackage[font=small,labelfont=bf]{caption}
\usepackage{multirow}
\usepackage{tikz}
\usetikzlibrary{calc}
\usepackage{pifont}%
\newcommand{\parbreak}{ \vspace{0.3 em}}
\usepackage{physics}

\usepackage{xcolor}

\usepackage{tabularx}

\usepackage{tabularx}

\newcounter{protocol}
\renewcommand{\theprotocol}{\arabic{protocol}}

\newenvironment{protocol}[1]
  {\par\addvspace{\topsep}
   \noindent
   \tabularx{\linewidth}{@{} X @{}}
    \hline
    \refstepcounter{protocol}\textbf{Protocol \theprotocol} #1 \\
    \hline}
  { \\
    \hline
   \endtabularx
   \par\addvspace{\topsep}}

\usepackage{graphicx}
\usepackage{subfig}
\usepackage{hyperref} 
\hypersetup{
    colorlinks,
    linkcolor={blue!80!black},%
    citecolor={green!30!black},
    urlcolor={blue!80!black}
}

\begin{document}

\title{Network Oblivious Transfer via Noisy Broadcast Channels}

\author{
Hadi Aghaee\inst{1}
\and Christian Deppe\inst{1}
\and Holger Boche\inst{2,3}
}

\institute{
Institute for Communications Technology,
Technische Universit\"at Braunschweig,
Braunschweig, Germany
\and
Chair of Theoretical Information Technology,
Technische Universit\"at M\"unchen,
Munich, Germany
\and
Munich Quantum Valley (MQV) and the Munich Center for Quantum Science and Technology (MCQST)\\
\email{\email{\{hadi.aghaee, christian.deppe\}@tu-bs.de}, boche@tum.de}
}

\maketitle

\begin{abstract}
Oblivious transfer is a fundamental cryptographic primitive whose information-theoretic realization over noisy channels has been extensively studied in the point-to-point setting. In contrast, little is known about oblivious transfer over broadcast channels, where a sender simultaneously communicates with multiple receivers and security requirements become substantially more involved. In this paper, we investigate information-theoretic oblivious transfer over discrete memoryless broadcast channels with one sender and two receivers. We consider both non-colluding and colluding honest-but-curious user models and derive general converse bounds on the achievable oblivious transfer capacity region for each scenario. Furthermore, we propose two explicit oblivious transfer protocols for binary erasure broadcast channels. The first protocol achieves the oblivious transfer capacity region for independent non-colluding receivers, thereby providing a complete characterization of the achievable rate region in this setting. The second protocol remains secure even in the presence of colluding receivers and establishes new achievable rate regions under stronger security requirements. Our results provide a unified information-theoretic framework for oblivious transfer over broadcast channels and demonstrate how the network topology fundamentally influences achievable oblivious transfer rates and security guarantees.
\end{abstract}


\section{Introduction}
In modern communication systems, the broadcast channel plays a crucial role by enabling a sender to simultaneously transmit messages to multiple receivers. This shared medium is especially important in wireless networks, satellite communications, and public messaging systems, where efficiency and synchronization across users are essential. However, the openness of the broadcast channel introduces inherent security challenges: since all users can potentially receive the transmitted data, ensuring confidentiality, integrity, and controlled access becomes vital. Secure communication over broadcast channels often requires additional cryptographic or information-theoretic mechanisms to prevent unauthorized parties from extracting sensitive information \cite{Korner}. 

\parbreak
Oblivious Transfer (OT) as a cryptographic primitive was first introduced by Rabin \cite{Rabin1} and developed in subsequent groundbreaking works \cite{Kilian1, EGL}. A well-known definition is that a sender (Alice) sends two secret bits to a receiver (Bob). Bob is allowed to get only one bit and should be unaware of the other bit, while Alice remains unaware of which bit Bob received. A more general OT scheme could be defined when Alice possesses $m$ secret strings of length $k$. It is shown that OT is complete for passive and active adversaries, so that if OT is possible between two parties, then they can securely compute any function between themselves \cite{Kilian1}. It is also shown that noisy channels are valuable resources for cryptographic purposes \cite{Joe}. Within this context, OT over broadcast channels emerges as a fundamental cryptographic primitive that enables secure two-party communication in a shared environment. By exploiting the properties of the broadcast channel, such as noise, erasures, or differential channel conditions, OT protocols allow a sender to transmit multiple pieces of information while ensuring that the receiver can access only one piece without revealing their choice, and the sender remains unaware of which piece was obtained. This fusion of broadcasting and privacy-preserving computation highlights the deep interplay between communication theory and cryptographic security.

Recently, the information-theoretic study of oblivious transfer has also been extended beyond point-to-point channels to more general network communication models. In particular, oblivious transfer over noisy multiple access channels has been investigated under different security assumptions, including impossibility results in the presence of non-signaling correlations as well as achievable constructions for colluding honest-but-curious users \cite{Hadi-ISIT26,Hadi-EW26}. These works demonstrate that the underlying network topology and the available communication resources have a fundamental impact on the feasibility and achievable rates of oblivious transfer. In this paper, we complement these results by investigating oblivious transfer over noisy broadcast channels.

\parbreak
\parbreak
Shannon’s assumption that an adversary receives exactly the same message as the legitimate receiver stems from the idealization of error-free communication channels \cite{Shannon}. In practice, however, most communication channels are inherently noisy. It is primarily for application purposes that these noisy channels are transformed into nearly error-free ones—albeit with reduced information rates—through the use of error-correcting codes. This seemingly minor observation implies that Shannon’s assumption may be overly restrictive when the underlying noisy channels themselves are available for cryptographic use. Building on this idea, Wyner \cite{Wyner} proposed a communication model in which Alice transmits information to Bob over a discrete memoryless channel, while an eavesdropper, Eve, observes Bob’s channel output only through an additional, independent, degraded channel. Wyner demonstrated that in this (admittedly idealized) scenario, Alice can achieve near-perfect secrecy in communication with Bob, even without a pre-shared secret key.

\parbreak
In cryptography, information-theoretic (or unconditional) security is generally considered more advantageous than computational security for two main reasons. Firstly, it does not rely on any assumptions regarding the adversary’s computational capabilities. Secondly, perfect secrecy represents the most robust notion of security, thereby eliminating the need to justify weaker alternatives, such as merely limiting an adversary’s probability of correctly guessing bits of plaintext. Although perfect secrecy is often dismissed as impractical due to Shannon’s pessimistic inequality $H(K)\geq H(M)$ (perfect secrecy is attainable only if the secret key is at least as long as the plaintext message), it is shown that OT with perfect secrecy is not attainable if one of the players is an unbounded cheater \cite{Hadi-limit}. However, we only focus on the case where all players are honest-but-curious, meaning that they do not have active attacks.

In contrast to the point-to-point setting, a broadcast channel simultaneously serves multiple receivers whose observations are statistically dependent. This creates new challenges for oblivious transfer, since security and reliability have to be guaranteed for multiple users while accounting for possible collusion. To the best of our knowledge, a general information-theoretic characterization of OT over broadcast channels has not been available.

\parbreak This article is organized as follows: The semantic intuitions and motivations behind the problem are discussed in Section \ref{Sec: motiv}. Some seminal definitions are presented in Section \ref{Sec: pre}. Section \ref{Sec: rel} is devoted to related work. The system model and main results are presented in Sections \ref{Sec: sys} and \ref{Sec: main}, respectively. Finally, we have a brief discussion in Section \ref{Sec: conc}.

\section{Motivation}\label{Sec: motiv}
OT cannot be obtained from scratch. This is the main impossibility result for interactive cryptographic protocols that are based on \emph{indistinguishability}, such as bit commitment, zero knowledge, secret sharing, non-interactive zero-knowledge, and in a broader sense, secure two-party computation, even when only considering security against computationally efficient adversaries \cite{Dodis}. In other words, such protocols cannot be realized without additional assumptions or restrictions. In \cite{Dodis}, it is proved that OT, as one of the most important cryptographic primitives rooted in two-party secure computation, cannot be realized with imperfect randomness obtained from any imperfect entropy sources. More surprising is that this impossibility result is unconditional, meaning that it does not rely on any further assumption like bounding the cheating capability of a party by honest model assumptions, or any other computational assumptions. This impossibility is even valid in the presence of a slightly imperfect entropy source (\emph{Santha-Vazirani (SV) source} \cite{SV}) if that is the only available source of randomness. This general phenomenon states that virtually all cryptographic tasks requiring any form of privacy or secrecy are impossible to realize when relying on SV entropy sources, irrespective of the computational assumptions employed. For further information about technical details, we refer the reader to \cite[Lemma 3.1]{Dodis}. Information theory offers a valuable entropy source to compensate for a lack of entropy: \emph{Noisy channels}. 

\parbreak Although often regarded as a malicious phenomenon in traditional and modern communication systems, noisy channels can serve as valuable entropy sources for cryptographic applications. Although the overall behavior of noise in a communication channel can be tracked and identified through experimentation, its individual realizations remain unpredictable. It may be possible to identify the type of noise and its effect on a message sequence, but the same noise has a different effect on the same message, such that the probability of observing two identical noisy messages approaches zero as the length of the sequence goes to infinity. This confirms the main result of \cite{Dodis}, which says that if two functions $F$ and $G$ produce computationally indistinguishable outputs when feeding any slightly imperfect randomness entropy source as input $x$, then $F(x) = G(x)$ for almost all inputs $x$. If we consider noise as a function $F$ acting on an encoded sequence (SV sequence a slightly nonuniform sequence $\{0,1\}^n$) $x^n = x_1 x_2 \ldots x_n$ twice, then if output sequences (noisy observations) are identical, then function $F$ should be deterministic for almost all inputs $x^n$. This, however, contradicts the inherently random nature of noise.

\parbreak Assuming the presence of noisy channels opens many fascinating doors to us. This line of research began in the seminal works \cite{Rabin1, Joe}, studied in \cite{Rudolf1, Winter1, Winter3, Crepeau3}, and developed in \cite{Imai, Watanabe2,Watanabe3}. However, the problem of shared noisy channels is less studied in this direction. A shared noisy channel allows multiple users to simultaneously interact with a single communication medium, introducing correlated noise patterns that can be exploited to achieve information-theoretic security \cite{Hadi2,Hadi1}. In contrast, single noisy channels between each pair of parties require independent noise realizations. By leveraging a shared noisy channel, the system can naturally implement joint encoding and decoding strategies, reduce communication overhead, and enhance privacy guarantees, particularly in scenarios where multiple receivers must extract information from a common source without revealing their choices. Therefore, understanding and utilizing the structure of shared noisy channels is fundamental for extending oblivious transfer protocols from simple point-to-point settings to more general network models like broadcast channels and multiple access channels.


\parbreak

\section{Preliminaries}\label{Sec: pre}
We use the well-known notation of information theory.

Consider a DMC with a transition matrix 
$W=\lbrace W(y|x),x \in\mathcal{X}, y \in\mathcal{Y} \rbrace$. There are two assumptions:
\begin{enumerate}
    \item \emph{Free Resources}: Alice and Bob have unlimited computing power, independent local randomness, and access to a noiseless public communication channel for an unlimited number of rounds.
    \item \emph{Honest-but-Curious Model}: Both parties follow the protocol honestly but may use all available information to infer what they are supposed to remain ignorant about.
\end{enumerate} 

The general two-party protocol:
\begin{itemize}
    \item \emph{Initial Views:} Alice and Bob start with initial knowledge or views $U$ and $V$, respectively.
    \item \emph{Random Experiments:} Both parties generate random variables $M$ and $N$ independently of each other and $(U,V)$.
    
    \item \emph{Message Exchange:} Alice sends Bob a message $C_1$ as a function of $U$ and $M$. Bob responds with $C_2$, a function of $V$, $N$ and $C_1$.

    \item \emph{Alternating Messages:} In subsequent rounds, they alternately send messages $C_3, C_4,\dots, C_{2t}$, which are functions of their instantaneous views.
    \item \emph{Final Views:} At the end of the protocol, Alice's view is $(U,M,\mathbf{C})$ and Bob's view is $(V,N,\mathbf{C})$, where $\mathbf{C}=C_1,\dots, C_{2t}$.

\end{itemize}
In the source model, Alice starts with $(M_0,M_1,X^n)$ and Bob with $(Z,Y^n)$ where $(M_0,M_1)$ are binary strings uniformly distributed on $\lbrace0,1\rbrace^k$, and $Z\in \lbrace 0,1\rbrace$ is a binary bit.
In the channel model, Alice starts with $(M_0,M_1)$ and Bob with $Z$. The structure of the protocol is the same as in the source model.

\subsection*{1. Notation for Tuples}
\begin{itemize}
    \item For a sequence $\mathbf{X}$, we define
    \begin{align*}
        \mathbf{X} &= X^n = (X_1, X_2, \dots, X_n),\\
        \mathbf{x} &= x^n = (x_1, x_2, \dots, x_n).
    \end{align*}
    
    \item Suppose $\mathcal{A}\subset \mathbb{N}$. Then $(\mathcal{A})$ represents the tuple formed by arranging the elements of $\mathcal{A}$ in increasing order:
    \[
    (\mathcal{A}) = (a_i |\, a_i \in A, \, i = 1, 2, \dots, |\mathcal{A}|), \quad \text{with } a_i < a_{i+1} \text{ for } i \geq 1.
    \]
    \textit{Example}: If $\mathcal{A} = \{1, 3, 2, 9, 4\}$, then $(\mathcal{A}) = (1, 2, 3, 4, 9)$.
    
    \item For a set $\mathcal{A} \subset \{1, 2, \dots, k\}$, and a sequence $\mathbf{X}$, we have:
    \[
    \mathbf{X}|_{\mathcal{A}} = \mathbf{X}|_{(\mathcal{A})} = \big\{ x_i|\,i\in \mathcal{A}\big\}.
    \]
    \textit{Example}: If $\mathbf{X} = (a, b, c, d, e, f, g)$ and $\mathcal{A} = \{6, 3, 1\}$, then $\mathbf{X}|_{\mathcal{A}} = (a, c, f)$.
    \item When an element $i$ is removed from a set $\mathcal{F}$, we write: 
    \[
    \mathcal{F}\setminus\{i\}.
    \]
    \textit{Example:} Given $\mathcal{F}=\{a, b, c, d,e\}$, we have: $ \{a, b, d,e\}=\mathcal{F}\setminus\{c\}$.
\end{itemize}

\subsection*{2. Markov Chains}
Random variables $X, Y, Z$ form a Markov chain $X - Y - Z$ when $X$ and $Z$ are conditionally independent given $Y$. That is, if $X \in \mathcal{X}$, $Y \in \mathcal{Y}$, and $Z \in \mathcal{Z}$, then $X - Y - Z$ implies:
\[
\forall x \in \mathcal{X}, \, \forall y \in \mathcal{Y}, \, \forall z \in \mathcal{Z}: \quad P_{X,Z|Y}(x,z|y) = P_{X|Y}(x|y) \cdot P_{Z|Y}(z|y).
\]
\subsection*{3. Erasure Count Function}
Given a sequence $y \in \{0, 1, e\}^n$, where $e$ indicates an erasure. We denote the erasure count function by:
\begin{align*}
    \Delta(y^n) &= |\{i \in \{1, 2, \dots, n\} : y_i = e \}|, \\
    \overline{\Delta}(y^n) &= |\{i \in \{1, 2, \dots, n\} : y_i \neq e \}|,
\end{align*}
where $y_i$ is a realization of $Y$.
\subsection*{4. Information Theoretic Definitions}

The min-entropy of a discrete random variable $X $ is 
\[
H_\infty(X) = \min_x \log \left( \frac{1}{P_X(x)} \right).
\]

Its conditional version is 
\[
H_\infty(X|Y) = \min_y H_\infty(X|Y=y).
\]

The zero-entropy and its conditional version are defined as 
\[
H_0(X) = \log |\{ x \in \mathcal{X} : P_X(x) > 0 \}|,
\]
and 
\[
H_0(X|Y) = \max_y H_0(X|Y=y).
\]

The statistical distance over two probability distributions $P_X $ and $P_Y $, defined over the same domain $\mathcal{X} $, is 
\[
\| P_X - P_Y \| = \frac{1}{2} \sum_{x \in \mathcal{X}} |P_X(x) - P_Y(x)|.
\]

For $\varepsilon \geq 0 $, the $\varepsilon $-smooth min entropy is 
\[
H_\infty^\varepsilon(X) = \max_{X' : \| P_{X'} - P_X \| \leq \varepsilon} H_\infty(X').
\]

Similarly,
\[
H_\infty^\varepsilon(X|Y) = \max_{X'Y' : \| P_{X'Y'} - P_{XY} \| \leq \varepsilon} H_\infty(X'|Y').
\]

Let $P_{UVW}$ be a probability distribution over $\mathcal{U}\times\mathcal{V}\times\mathcal{W}$ For any $\varepsilon > 0 $ and $\varepsilon' > 0$ it holds that \cite{Holenstein}: 

\begin{equation}\label{eq: holenstein-1}
    H_\infty^{\varepsilon+\varepsilon'}(U|V,W) \geq H_\infty(U|W) + H_\infty^\varepsilon(V|U,W) - H_0(V|W) - \log(\frac{1}{\varepsilon'}).
\end{equation}
\parbreak Also, $H^{\varepsilon}_{\infty}(U,V|W)$ can be bounded from below and above as \cite{Holenstein}: 

\begin{align}\label{eq: holenstein-2}
    H_\infty^{\varepsilon+\varepsilon'}(U|V,W) + H_{0}(V|W) + \log(\frac{1}{\varepsilon'})&\leq H^{\varepsilon}_{\infty}(U,V|W)\notag\\
    & \leq H_{\infty}(U|W) + H^{\varepsilon}_{\infty}(V|U,W).
\end{align}

Combining \eqref{eq: holenstein-1} and \eqref{eq: holenstein-2}, yields: 
\begin{equation}\label{eq: holenstein-1 + holenstein-1}
    H^{\varepsilon}_{\infty}(U,V|W)\geq H_\infty(U|W) + H_\infty^\varepsilon(V|U,W). 
\end{equation}
\begin{lemma}\cite{Tomamichel2}\label{lemm: min-entropy vs smooth}
For finite classical random variables $U,V,W$ and small smoothing parameters
$\varepsilon,\varepsilon',\varepsilon''\ge 0, \varepsilon>\varepsilon'+2\varepsilon''$,
\begin{equation}\label{eqeq}
H_{\infty}^{\varepsilon+\varepsilon'+2\varepsilon''}(UV\mid W)
\ge
H_{\infty}^{\varepsilon'}(U\mid W)
+
H_{\infty}^{\varepsilon''}(V\mid UW)
-f,
\end{equation}
where $f$ is at most of order $O(\log \frac{1}{\gamma})$ when $\gamma = \varepsilon-\varepsilon'-2\varepsilon''$ is small. 
\end{lemma}

\begin{proof}
This is the classical specialization of the standard smooth min-entropy
chain rule in the quantum setting; see, e.g., Tomamichel's monograph on
quantum information processing with finite resources (smooth entropy
calculus) \cite{Tomamichel1,Tomamichel2}. Since a classical distribution is a diagonal density operator,
the quantum statement applies verbatim.
\end{proof}
We use the notation $\overline{A} \triangleq A \oplus 1$ to denote the binary complement of a bit $A \in \{0,1\}$, where $\oplus$ denotes the bitwise exclusive OR (XOR) operation. This convention is used throughout the paper for all binary random variables. For non-binary random variables, we mean the complement component of their sets. Example: $i \in \{1,2\}$, if $i = 2$ then $\overline{i}=1$.

\begin{definition}\label{Renyi entropy}
    Given a random variable $X$ with alphabet $\mathcal{X}$ and probability distribution $p_X$, the \textit{Rényi entropy of order two} of a random variable $X$ is given by:
    \[
    H_2 (X) = \log_2\left(\frac{1}{P_c(X)}\right).
    \]
    where the \textit{collision probability} $P_c(X)$ is the probability that two independent trials of $X$ produce the same outcome. It is defined as:
    \[
    P_c(X) = \sum_{x \in \mathcal{X}} p_X(x)^2.
    \]
    For a given event $\mathcal{E}$, the conditional distribution $p_{X|\mathcal{E}}$ is employed to define the \textit{conditional collision probability} $P_c(X|\mathcal{E})$ and the \textit{conditional Rényi entropy of order 2}, $H_2(X|\mathcal{E})$.
\end{definition}
\begin{lemma}\cite[Corollary 4]{Bennet}\label{entropy hash}
Let $p_{X,Y}$ be any probability distribution, where $X \in \mathcal{X}$, $Y \in \mathcal{Y}$, and $y$ is a specific element of $Y$. Assume that $H_2(X|Y = y) \geq c$ for some constant $c$. Let $\mathcal{K}$ be a universal class of functions mapping $X$ to $\{0,1\}^l$, and let $\kappa$ be sampled uniformly from $\mathcal{K}$. Then:
\begin{align*}
    H(\kappa(X)|\kappa, Y = y) & \geq l - \log(1 + 2^{l-c})\\
    & \geq l - \frac{2^{l-c}}{\ln 2}.
\end{align*}
\end{lemma} 
\begin{lemma}\cite[
Corollary 2.12]{Holenstein}\label{lemm: Holenstein}
    Let $P_{X^nY^n}$ be independent and identically distributed (i.i.d.) according to $P_{XY}$ over the alphabet $\mathcal{X}^n \times \mathcal{Y}^n$. For any $\varepsilon > 0$, we have
\[
H_\infty^\varepsilon(X^n | Y^n) \geq nH(X | Y) - 4\sqrt{n \log(1/\varepsilon)}\log|\mathcal{X}|.
\]
\end{lemma}
\begin{definition}\label{Hash}
A function $h: \mathcal{R} \times \mathcal{X} \to \{0,1\}^n$ is a \textit{two-universal hash function} \cite{Carter} if, for any $x_0 \neq x_1 \in \mathcal{X}$ and for $R$ uniformly distributed over $\mathcal{R}$, it holds that
\begin{equation}\label{eq: hashs}
    \Pr(h(R, x_0) = h(R, x_1)) \leq 2^{-n}.
\end{equation}

Similarly, given two independent hash functions $h_1: \mathcal{R} \times \mathcal{X} \to \{0,1\}^n$ and $h_2: \mathcal{T} \times \mathcal{Y} \to \{0,1\}^m$, for any $x_0 \neq x_1 \in \mathcal{X}$, $y_0 \neq y_1 \in \mathcal{Y}$, and for $R, T$ uniformly distributed over $\mathcal{R}$ and $\mathcal{T}$, respectively, it holds that
\begin{equation}\label{eq: hashs-double}
    \Pr(h_1(R, x_0) = h_1(R, x_1)\cap h_2(T, y_0) = h_2(T, y_1)) \leq 2^{-(n+m)}.
\end{equation}

 An example of a two-universal class is the set of all linear mappings from $\{0,1\}^n$ to $\{0,1\}^r$. 
 \end{definition}
\begin{lemma} \cite{Winter1, Jurg} (Distributed leftover hash lemma)\label{DLHL}
    Let $\varepsilon > 0$, $\varepsilon' \geq 0$, and let $g_i: \mathcal{T}_i \times \mathcal{X}_i \to \{0, 1\}^{n_i}$ for $1 \leq i \leq m$ be two-universal hash functions. Let $X_i$ be random variables taking values in $\mathcal{X}_i$, $1 \leq i \leq m$, where for any subset $\mathcal{S} \subseteq \{1, 2, \dots, m\}$, and $\mathbf{X}|_\mathcal{S} = X_{\mathcal{S}(1)}, X_{\mathcal{S}(2)}, \dots, X_{\mathcal{S}(|\mathcal{S}|)}$, we have
\[
H_\infty^{\varepsilon'}(\mathbf{X}|_\mathcal{S} | Z) \geq \sum_{i \in \mathcal{S}} n_i + 2 \log(1/\varepsilon),
\]
where $T_1, \dots, T_m$ are uniformly distributed over $\mathcal{T}_1 \times \dots \times \mathcal{T}_m$, and are independent of $X_1, \dots, X_m$, and $Z$. Here, \(Z\) denotes arbitrary side information jointly distributed with \(X_1,\dots,X_m\), and independent of the hash seeds \(T_1,\dots,T_m\). Then, it holds that the tuple \[
(g_1(T_1, X_1), \dots, g_m(T_m, X_m))\] is $(2^m \varepsilon / 2 + 2^m \varepsilon')$-close to uniform with respect to $T_1, \dots, T_m, Z$.
\end{lemma}

\begin{lemma}\cite[Lemma 3]{Rudolf1}\label{lemma: Rudolf}
    Let $X$, $Y$, and $Z$ be random variables defined on the finite sets $\mathcal{X}$, $\mathcal{Y}$, and $\mathcal{Z}$, respectively. For any $z_1, z_2 \in \mathcal{Z}$ with $p \triangleq \mathbb{P}[Z = z_1] > 0$ and $q \triangleq \mathbb{P}[Z = z_2] > 0$, the following inequality holds:
\[
\left\lvert H(X | Y, Z = z_1) - H(X | Y, Z = z_2) \right\rvert \leq 1 + 3 \log |\mathcal{X}| \sqrt{\frac{(p + q) \ln 2}{2pq} I(X,Y; Z)}.
\]
\end{lemma}

\parbreak 
\section{Related Work}\label{Sec: rel}
Nascimento and Winter reduced a general noisy correlation (point-to-point channel) to a slightly unfair noisy symmetric basic correlation (SU-SBC). They also showed that any imperfect noisy point-to-point channel/correlation could be used to obtain a certain slightly unfair noisy channel/correlation (SUNC/SUCO), and any SUNC/SUCO could be used to implement a certain SU-SBC. Finally, they proved that in such a reduced model, one can implement $\binom{m}{1}-\text{OT}^k$ (sending $m$ secret sequences of length $k$ by a sender and choosing only one message by a receiver) at a positive rate if Alice (sender) is honest-but-curious \cite{Winter1,Winter3}. These papers are important for the following reasons: 1- The concept of oblivious transfer capacity of a Discrete Memoryless Channel (DMC) is defined for the first time as the supremum of all achievable rates $R$ such that $\frac{k}{n}\geq R - \gamma, \gamma> 0$, where $n$ is the number of channel uses. 2- They also considered the problem in a malicious model where a malicious player could arbitrarily deviate from the prescribed channel statistics in only $\delta n$ positions; otherwise, he/she will be caught by the other player with a certain probability. Ahlswede and Csisz\'ar studied the same problem in the honest-but-curious model without reducing the general model \cite{Rudolf1}. They calculated an upper bound on the OT capacity of a point-to-point DMC that matches the OT-capacity bound of Nascimento and Winter. For a general noisy point-to-point DMC reduced to a special erasure channel, they also proved a lower bound on OT capacity. In short, the OT capacity of a noisy DMC is unknown in general. In fact, for general channels, we have to first convert the channel to a Generalized Erasure Channel (GEC) via alphabet extension and erasure emulation and then apply the general construction for GEC.

More recently, oblivious transfer has been investigated in network communication settings beyond point-to-point channels. Information-theoretic impossibility results for oblivious transfer over noisy multiple access channels in the presence of non-signaling correlations were established in \cite{Hadi-ISIT26}. Furthermore, achievable oblivious transfer protocols for noisy multiple access channels with colluding honest-but-curious users were proposed in \cite{Hadi-EW26}. These contributions highlight the influence of network topology and adversarial assumptions on the achievable OT capacity and motivate the present investigation of broadcast channels.

\parbreak In the following, we present the bounds achieved by Ahlswede and Csiszár:

\begin{itemize}
    \item Upper bound: The OT capacity of Discrete Memoryless Multiple Source (DMMS) with generic random variables $X, Y$ or of a DMC $\lbrace W \rbrace$ is bounded above by 

    \begin{equation}\label{Upper}
        C_{\text{OT}}\leq\,\min \Big[\max_{P_X}I(X;Y), \max_{P_X}H(X|Y)\Big].
    \end{equation}
    
    \item Lower bound: Consider the following form for a DMC $\lbrace W \rbrace$:
    \[
    W(y|x)= \begin{cases}
             (1-p)\,W_0(y|x)  & x\in\mathcal{X}, y\in \mathcal{Y}_0\\
             p\,W_1(y|x)  & x\in\mathcal{X}, y\in \mathcal{Y}_1
       \end{cases},
    \]
   
    where outputs $y\in\mathcal{Y}_1$ carry no information about the input and they are interpreted as erasures.
    For such a case and any distribution $P$ on $\mathcal{X}$
    \begin{equation}\label{Lower}
        C_{\text{OT}}\geq \left[I(U;Y^{(0)})-I(U;Y^{(1)})\right].\min(p, 1-p),
    \end{equation}
    where $U$ is any random variable and $X, Y^{(0)}, Y^{(1)}$ are random variables with $P_{XY^{(j)}}(x,y) = P(X)W_{j}(y|x)$, $j=0, 1$, such that $U\rightarrow X \rightarrow (Y^{(0)}, Y^{(1)})$ is a Markov chain. This means that the OT capacity is bounded below by $\min (p, 1-p)$ times the secrecy capacity of the wiretap channel with component channels $W_0, W_1$. 
\end{itemize}
\parbreak
The structure of Even, Goldreich, and Lempel (EGL) oblivious transfer protocol \cite{EGL} is widely used in information theory \cite{Winter1, Winter3, Rudolf1, Imai}: Consider the $1$-out-of-$2$ version of OT. Alice possesses two random variables $x_0$ and $x_1$ corresponding to wired labels (messages) $m_0$ and $m_1$, respectively. First, she shares an RSA key with Bob over a public channel (Key$_{public}=K_{A_1}(e,n)$). Then she sends her messages over a noisy channel to Bob. Bob also enters a bit $b\in \lbrace 0,1 \rbrace$ into the channel. So his choice will be $x_b$. Then he encrypts his chosen message using the public RSA key and chooses a random value $(r)$ as $\nu=x_b+r^r \mod n$ and sends it to Alice. Then Alice tries to decrypt Bob's message using her private key (Key$_{private}=K_{A_2}(d,n)$). If Bob's choice is $x_1$, then Alice finds $k_1=(\nu-x_1)^d \mod n = r$ and obtains no useful information about $k_0=(\nu-x_1-x_0)^d \mod n$. This simply shows that the protocol is secure for Bob. Then, Alice sends two new wired labels $m'_0=m_0+k_0$ and  $m'_1=m_1+k_1$ to Bob. He also cannot learn more than $x_b=x_1$. This shows that the protocol is secure for Alice. 

\section{System model}\label{Sec: sys}
\subsection{OT over a Discrete Memoryless Broadcast Channel (DM-BC)}
A DM-BC $(\mathcal{Y}_1\times \mathcal{Y}_2, p_{Y_1,Y_2|X}, \mathcal{X})$, where $X\in\mathcal{X}\triangleq \{0,1\}$, and $Y_i\in\mathcal{Y}_i\triangleq\{0,1\}, i\in\{1,2\}$ has one sender and two or more receivers. Here, we consider the simplest form of DM-BC with one sender (Alice) and two receivers (Bob-1, Bob-2). Alice has four strings/messages $(m_{i0},m_{i1}), i\in \{1,2\}$ where $M_{i0},M_{i1}$ are independently and uniformly distributed over $\{0,1\}^{k_i}$. Alice wants to implement an $\binom{2}{1}-\text{OT}^{k_i}$ protocol with Bob-$i$, $i\in\{1,2\}$. Alice sends $(m_{i0},m_{i1})$ over the noisy broadcast channel to Bob-$i$, and Bob-$i$ has to choose one of them by inputting to the channel $Z_i\in\{0,1\}$. Bob-$i$ should be unaware of the unselected message $m_{i\overline{Z}_i}$ while Alice is unaware of the chosen message by Bob-$i$ $(m_{iZ_i})$. At the end of the protocol, Bob-1 and Bob-2 get $m_{1Z_1}$ and $m_{2Z_2}$, respectively, while Alice gets nothing.

\parbreak 
This structure was first studied in \cite{mishra} wherein the DM-BC is made up of two independent BECs: BEC($p_i$) from Alice to Bob-$i$, $i\in\{1,2\}$ where $p_i$ is the erasure probability. As is apparent, this channel model is restricted and considers only a traditional model of noisy channels based on BECs. Here, we aim to investigate the OT problem over a general DM-BC with honest-but-curious parties. 

\parbreak Consider Fig. \ref{fig: both}-$(a)$. First, general multi-party protocols based on the DM-BC should be described. 

\subsubsection{Noiseless Protocol}

\parbreak Assume Alice, Bob-1, and Bob-2 possess initial knowledge or views $U', V'_1$, and $V'_2$, respectively, where $U', V'_1$, and $V'_2$ are not necessarily independent random variables. At the beginning of the protocol, the players independently perform random experiments to generate random variables $R_{A}, R_{B_1}$, and $R_{B_2}$, respectively, such that $(R_{A}, R_{B_1}, R_{B_2})$ and $(U', V'_1, V'_2)$ are mutually independent. Alice then sends a message $C_{i,1},i\in\{1,2\}$, which is a function of $U'$ and $R_{A}$, to Bob-$i$ over the noiseless public channel. Bob-$i$ responds with message $C_{i,2}$ each to Alice as a function of $V'_i, R_{B_i}$ and $C_{i,1}$. The random variables $R_{A}$ and $R_{B_i}$ model the randomization in Alice's and Bob-$i$'s choices of $C_{i,1}$ and $C_{i,2}$ as well as in their subsequent actions. 

\parbreak In subsequent rounds, Alice and Bob-$i$ alternately send messages \[ C_{i,3}, C_{i,4}, \cdots, C_{i,2t}\] where each $C_{i,K}$ depends on their current view. Specifically $C_{i,K}$ is a function of $U', R_{A}$ and $\{C_{i,j}, j\leq K\}$ if $j$ is odd and of $V'_i, R_{B_i}$ and $\{C_{i,j}, j\leq K\}$ if $j$ is even. At the end of the protocol, Alice's view is $(U', R_{A}, \mathbf{C}_i)$ where $\mathbf{C}_i$ denotes the concatenation of all messages exchanged $C_{i,1}, \cdots, C_{i,2t}$ over the noiseless public channel.

\parbreak For this discussion, Alice and Bob-$i$ are assumed to have the following resources available for free: (i) Unlimited computational power. (ii) Local randomness from independent random experiments. (iii) A noiseless public channel for unlimited communication in any number of (finite) rounds. These resources alone, however, are insufficient for OT. Two models incorporating additional resources can be defined:

\begin{enumerate}
    \item \emph{Source Model}: Defined by a Discrete Memoryless Multiple Sources (DMMS) consisting of i.i.d. repetitions $(X_{l}, Y_{1,l}, Y_{2,l}),\\ l = \{1, 2, \cdots, n\}$ of generic random variables $(X, Y_1, Y_2)$ taking values in finite sets $\mathcal{X}$, $\mathcal{Y}_1$, and $\mathcal{Y}_2$. Alice observes $X_{l}$ and Bob-$i$ observes $Y_{i,l}$ at the $l$-th access.
    \item \emph{Channel Model}: Defined by a DM-BC with finite input and output alphabets $\mathcal{X}$, $\mathcal{Y}_1$, and $\mathcal{Y}_2$ where the conditional probability of Bob-1 and Bob-2 receiving $y_1, y_2$, respectively, when Alice sends $x$, is $W(y_1,y_2|x)$. Alice selects an input $X_{l}$, and Bob-$i$ observes the corresponding output $Y_{i,l}$ at the $l$-th access.
\end{enumerate}

\parbreak The cost of an OT protocol in either model is the number of accesses to the DMMS and DMC. The total capacity $C_{\text{OT}}$ of a DMMS or DMC is defined as the limit of $\frac{1}{n}$ times the largest $k_i$ for which OT is possible with cost $n$, as $n\rightarrow \infty$. 

\subsubsection{Noisy Protocol}\label{noisy protocol:BC}
    \parbreak A noisy protocol utilizing $n$ accesses to DM-BC ($\mathcal{W}: \mathcal{X}\rightarrow \mathcal{Y}_1\mathcal{Y}_2$) operates as follows: Alice and Bob-$i$ start with their initial views represented by random variables $U'$ and $V_i$, respectively, and generate random variables $R_{A}$ and $R_{B_i}$ as previously described. Alice then chooses the input $X$ for the DM-BC as a function of $U'$ and $R_{A}$, and Bob-$i$ observes the corresponding output $Y_i$. 

\parbreak Subsequently, they engage in a session of public communication following a noiseless protocol, where their initial views are updated to $(U', R_{A})$ and \[(V'_i, R_{B_i}, Y_i),\] respectively. Note that $X_{1}$ does not need to be part of Alice's view since it is already determined by $U', R_{A}$. In the course of this and any subsequent public communication sessions, Alice and Bob-$i$ use the original random variables $R_{A}$ and $R_{B_i}$ for all necessary randomization. DM-BC accesses alternate with public communication sessions. Let $C_i^l$ represent the cumulative public communication up to the $l$-th session. Before the $l$-th DM-BC access, Alice's view is $(U', R_{A}, C_i^{l-1})$. Alice selects the DM-BC input $X_{l}$ based on their views, and Bob-$i$ observes the output $Y_{i,l}$. Conditioned on $X_{l} = x$, the random variable $Y_{i,l}$ is independent of $U', V'_i, R_{A}, R_{B_i}, Y_i^{l-1}, C_i^{l-1}$, and follows the distribution $W(.\,,.|x)$. During the $l$-th public communication session, Alice and Bob-$i$ execute a noiseless protocol where their views are $(U', R_{A}, C_i^{l-1})$ and $V'_i, R_{B_i}, Y_i^l, C_1^{l-1},C_2^{l-1}$, respectively. The protocol concludes after $n$ public sessions, with Alice's final view as $(U, R_{A},\mathbf{C}_i)$ and Bob's as $(V_i, R_{B_i}, Y_i^n, \mathbf{C}_1 \mathbf{C}_2)$, where $\mathbf{C}_i = C_i^n$ represents the total public discussion between Alice and Bob-$i$. Alice's knowledge of $X^n = (X_{1}, \cdots, X_{n})$ is implicit as it is a function of $(U', R_{A}, \mathbf{C}_i)$. 

\parbreak Admissible protocols for the cost-$n$ oblivious transfer of length-$k_i$ messages, or briefly $(n,k_1,k_2)$ OT protocols, are defined as follows. Here, \[X^n = (X_{1}, \cdots, X_{n})\] and $Y_i^n = (Y_{i,1}, \cdots, Y_{i,n})$ represent, in the source models, the source output sequences observed by Alice and Bob-$i$, and in the channel model, the sequences of DM-BC inputs selected by Alice and outputs observed by Bob-$i$. In the source model, Alice and Bob-$i$ can use any noiseless protocol where their initial views are $U' = (M_{i0}, M_{i1}, X^n)$ and $V'_i = (Z_1, Z_2, Y_i^n)$. Here $M_{i0}$ and $M_{i1}$, representing the two binary strings given to Alice, are uniformly distributed over $\{0,1\}^{k_i}$, and $Z_i$ representing the bits given to Bob, is uniformly distributed over $\{0,1\}$. $M_{i0}, M_{i1}, M_{iZ_i}, (X^n, Y_i^n)$ are mutually independent. In the channel model, any noisy protocol with $n$ accesses to the DM-BC is permissible, where the initial views are $U' = (M_{i0}, M_{i1})$ and $V'_i = Z_i$ with $M_{i0}, M_{i1}$ uniformly distributed over $\{0,1\}^{k_i}$, and $Z_i$. In both models, Bob-$i$ produces an estimate $\hat{M}_{iZ_i}$ of $M_{iZ_i}$ as a function of his view $(Z_i, R_{B_i}, Y_i^n, \mathbf{C}_i)$ upon completing the protocol.

\subsubsection{Collusion scenarios}
\parbreak
We consider two scenarios: i) colluding parties and ii) non-colluding parties.
If none of the players collude with others, then the general secrecy criteria for OT links should be satisfied. If the receivers are colluding parties, then Bob-1 can collude with Bob-2 and Alice, or vice versa. So, Bob-$\overline{i}$ is a wiretapper from the perspective of the OT link between Alice and Bob-$i$.  

\parbreak
\begin{definition}\label{def: main}
    Let $n, k_1, k_2 \in \mathbb{N}$. An $(n, k_1, k_2)$ protocol is an interaction involving Alice, Bob-1, and Bob-2 through the setup shown in Fig. \ref{fig: both}-$(a)$. During each time step $l = 1, 2, \ldots, n$, Alice transmits a bit $X_l$ via the broadcast channel. Additionally, users alternately send messages over a noiseless public channel over several rounds before each transmission and after the last one $(l = n)$. The number of rounds may vary but is finite. Any transmission by a user is a function of their input, private randomness, and all previous messages, channel inputs, or channel outputs they have observed. The positive rate pair $(R_1,R_2)$ is an achievable OT rate pair for the DM-BC if there exists a sequence of protocols such that $\frac{k_i}{n}\rightarrow R_i$ as $n\to \infty$  such that: i) for non-colluding parties, the asymptotic conditions $\eqref{goals: BC-nColl-1}-\eqref{goals: BC-nColl-3}$ hold and ii) for colluding parties the asymptotic conditions $\eqref{goals: BC-Coll-1}-\eqref{goals: BC-Coll-5}$ hold. The set of all achievable rate pairs is the OT capacity region of the DM-BC. 
\end{definition}
\begin{figure}[t]
    \centering
    \begin{minipage}{0.45\textwidth}
        \centering
        \includegraphics[scale=0.71,trim={2cm 20.5cm 11cm 1.5cm},clip]{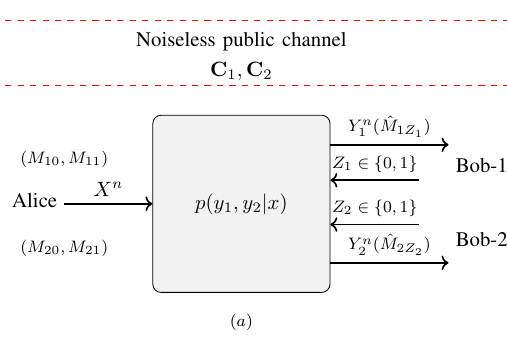}
    \end{minipage}
    \hfill
    \begin{minipage}{0.45\textwidth}
        \centering
        \includegraphics[scale=0.71,trim={2cm 20.5cm 10cm 1.5cm},clip]{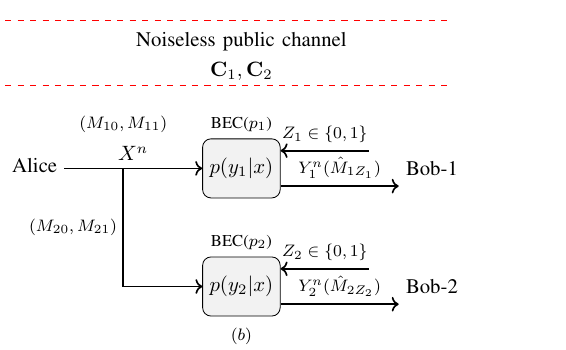}
        
    \end{minipage}
    \caption{$(a)$ OT over a general BC. $(b)$ OT over a DM-BC made up of two independent BECs. Note that $Z_1$ and $Z_2$ are not the inputs to the channel, but to the OT protocol and influence the public transmission $\mathbf{C}_1$, and $\mathbf{C}_2$.}
    \label{fig: both}
\end{figure}
Let $\mathbf{C}$ denote the transcript of the public channel after an $(n, k_1, k_2)$ protocol. The final view of a user includes all random variables available to them at the end of the \((n, k_1, k_2)\)-protocol. We represent the final views of Alice, Bob-1, and Bob-2 by $U$, $V_1$, and $V_2$, respectively. At the end of the $(n, k_1, k_2)$ protocol, Bob-$i$ derives an estimate $\hat{M}_{iZ_i}$ of $M_{iZ_i}$ based on his final view $V_i$.

\parbreak For \emph{non-colluding} parties and in the asymptotic limit, these goals are formalized by the following conditions:
\begin{align}\label{goals: BC-nColl-1}
    \lim_{n\rightarrow\infty}\text{Pr}\,\left[(\hat{M}_{iZ_i})\neq (M_{iZ_i})\right]_{i\in\{1,2\}}= 0,\\
    \label{goals: BC-nColl-2}
    \lim_{n\rightarrow\infty}I (  M_{i\overline{Z}_i};V_i)_{i\in\{1,2\}} = 0,\\
    \label{goals: BC-nColl-3}
    \lim_{n\rightarrow\infty}I (Z_1, Z_2 ; U) =  0,
\end{align}
where $V_i = (Z_i, R_{B_i}, Y_i^n, \mathbf{C}_i), U = (M_{i0}, M_{i1}, R_{A}, X^n , \mathbf{C}), i\in \{1,2\}$
and $\mathbf{C}\triangleq (\mathbf{C}_1, \mathbf{C}_2)$. Condition \eqref{goals: BC-nColl-1} says that Bob-$i$ correctly learns $M_{iZ_i}$ with negligible probability of error, so that the players do not abort the protocol (\emph{correctness}). Condition \eqref{goals: BC-nColl-2} says that Bob-$i$ gains negligible information about $M_{i\overline{Z}_i}$ (\emph{security for Alice}), and condition \eqref{goals: BC-nColl-3} says that Alice gains negligible information about $Z_1, Z_2$ (\emph{security for Bob-$i$}). 

\parbreak For \emph{Colluding} parties and in the asymptotic limit, the goals are formalized by the following conditions:
\begin{align}\label{goals: BC-Coll-1}
    \lim_{n\rightarrow\infty}\text{Pr}\,\left[(\hat{M}_{iZ_i})\neq (M_{iZ_i})\right]_{i\in\{1,2\}}&= 0,\\
    \label{goals: BC-Coll-2}
    \lim_{n\rightarrow\infty}I (  M_{1\overline{Z}_1},M_{2\overline{Z}_2};V_1,V_2)&= 0,\\
    \label{goals: BC-Coll-3}
    \lim_{n\rightarrow\infty}I (Z_i; U, V_{\overline{i}})_{i\in\{1,2\}} &=  0,\\\label{goals: BC-Coll-4}
    \lim_{n\rightarrow\infty}I (Z_1, Z_2 ; U) &=  0,\\\label{goals: BC-Coll-5}
    \lim_{n\rightarrow\infty}I (M_{i0},M_{i1},M_{\overline{i}\overline{Z}_{\overline{i}}}, Z_i; V_{\overline{i}})_{i\in\{1,2\}} &=  0,
\end{align}
where $V_i = (Z_i, R_{B_i}, Y_i^n, \mathbf{C}_i), U = (M_{i0}, M_{i1}, R_{A}, X^n , \mathbf{C}), i\in \{1,2\}$
and $\mathbf{C}\triangleq (\mathbf{C}_1, \mathbf{C}_2)$. Condition \eqref{goals: BC-Coll-1} says that Bob-$i$ correctly learns $M_{iZ_i}$ with negligible probability of error, so that the players do not abort the protocol (\emph{correctness}). Condition \eqref{goals: BC-Coll-2} says that Bob-$i$ gains negligible information about $M_{i\overline{Z}_i}$ (\textit{security for Alice}), and conditions \eqref{goals: BC-Coll-4} says that Alice gains negligible information about $Z_1, Z_2$ (\textit{security for both Bobs}). The other two criteria are related to the security of parties against collusion.

\begin{remark}\label{rem}
   As shown in \cite{Hadi-limit}, OT with perfect secrecy is impossible when one of the parties is an unbounded adversary. This observation motivates our focus on asymptotic limits. Nevertheless, Ahlswede and Csiszár demonstrated that for $\binom{2}{1}\text{-OT}^k$ over a point-to-point noisy DMC $(\mathcal{W}:\mathcal{X}\to\mathcal{Y})$ reducible to a binary erasure channel with erasure probability $p$, Alice’s information about Bob’s input $Z$ is not only asymptotically but exactly zero. Moreover, when $X$ is uniformly distributed over $\{0,1\}$, Bob’s information about $M_{\overline{Z}}$ is also exactly zero \cite[Remark~5]{Rudolf1}. This result suggests that perfect secrecy is achievable under the assumption that the parties behave honestly, which aligns with the impossibility results in \cite{Hadi-limit}. However, care must be taken when defining the adversarial model. Cryptographers typically classify the possibility of collusion as malicious behavior: a malicious party is always allowed to collude with others, and a colluding party is usually malicious, but not necessarily so. In this work, we adopt a fully honest model in which players may collude but are not permitted to deviate from the prescribed OT protocol. Under this definition of adversary capability, perfect OT can indeed be expected for honest parties, even if collusion is allowed.
\end{remark}
\begin{corollary}\label{cor}
    If an OT protocol $\Pi$ provides perfect secrecy for both parties, then $\Pi$ continues to provide perfect secrecy even in the presence of a passive wiretapper who may collude with other parties.
    
\end{corollary}

\begin{proof}
    The proof is presented in \ref{app: cor}.
\end{proof}
\section{Main results}\label{Sec: main}
In this section, we present the principal theoretical contributions of our work. We derive general outer bounds on the achievable OT rate region for both non-colluding and colluding honest-but-curious parties over a DM-BC. Furthermore, we propose two explicit OT protocols tailored for noisy broadcast environments: the first ensures correctness and privacy for non-colluding users, while the second maintains secrecy even in the presence of collusion. Together, these results provide a unified framework for understanding and quantifying the OT capacity in broadcast networks under information-theoretic security constraints.

\begin{theorem}\label{thm:BC-ncoll}
    In a general $\binom{2}{1}\emph{-OT}^{k_i}$ setup over a DM-BC $\mathcal{W}:\mathcal{X}\rightarrow \mathcal{Y}_1\mathcal{Y}_2:P_{Y_1Y_2|X}(y_1,y_2|x)$ with honest-but-curious non-colluding  parties (Fig. \ref{fig: both}-$(a)$), if an OT rate pair is achievable, then it belongs to the set of all rate pairs $(R_1,R_2)\in\mathbb{R}^2_{+}$ that satisfy:
    \begin{align*}
        R_1 &\leq \min\Big\{\max_{P_{X}}I(X;Y_1),\max_{P_{X}}H(X|Y_1)\Big\},\\
        R_2 &\leq \min\Big\{\max_{P_{X}}I(X;Y_2),\max_{P_{X}}H(X|Y_2)\Big\},\\
        R_1+R_2 &\leq \max_{P_X}I(X;Y_1,Y_2),
    \end{align*}
\end{theorem}
\begin{proof}[Proof sketch] The proof follows an extended version of the proof in \cite[Theorem 1]{Rudolf1}. In the DM-BC scenario described above, there is an OT link between Alice and Bob-1 and another OT link between Alice and Bob-2. So, applying the method in \cite[Theorem 1]{Rudolf1} for each OT link completes the proof. 
\end{proof}
\begin{proposition}
    The rate region $\mathcal{R}$ of independent pairs of $\binom{2}{1}\text{-OT}^{k_i}$ protocols over the DM-BC  (Fig. \ref{fig: both}-(a)): $\mathcal{W}:\mathcal{X}\rightarrow \mathcal{Y}_1\mathcal{Y}_2:P_{Y_1|X}(y_1|x)P_{Y_2|X}(y_2|x)$ with non-colluding parties is such that $\mathcal{R}\subseteq\mathcal{R}_{\text{outer}}$:
    \begin{align*}
    \mathcal{R}_{\text{outer}} = \left\{\ 
    \begin{aligned}
    (R_1,R_2)\in\mathbb{R}^2_{+}:R_1 &\leq \min\{p_1, 1 - p_1\}, \\
    R_2 &\leq \min\{p_2, 1 - p_2\}, \\
    R_1 + R_2 &\leq  1-p_1p_2
    \end{aligned}
    \right\}
    \end{align*}
\end{proposition}
\begin{protocol}{OT over noisy DM-BEBC with non-colluding parties in the setup of Fig. \ref{fig: both}-$(b)$.}\label{protocol: BC-nColl}
         \textit{Parameters:} \begin{itemize}
    \item $\lambda\in (0,1)$ such that $r_i<\min \{p_i,1-p_i\}-\lambda, i \in \{1,2\}$, where $p_i$ is the erasure probability.
    \item $0<\lambda'<r_i, \lambda'\in \mathbb{Q}$
    \item $n(r_i-\lambda')\in\mathbb{N}$
    \item The rate of the protocol is $r_i-\lambda' = \frac{k_i}{n}$, and $k_i$ is the length of each string.
\end{itemize}
\textit{Goal.} Alice sends two strings $m_{i0},m_{i1}, i \in \{1,2\}$ to Bob-$i$. At the end of the protocol, Bob-$i$ gets $m_{iZ_i}, Z_i\in \{0,1\}$ while Alice gets nothing ($e$).

\textit{The protocol:}
  \begin{enumerate}
    \item Alice transmits an $n$-tuple $\mathbf{X}= X^n$ of i.i.d. Bernoulli($\frac12$) bits over both point-to-point BECs($p_i$). 

    \item Bob-$i$ receives the $n$-tuple $\mathbf{Y_i}= Y_i^n$ from BEC($p_i$). Bob-$i$ forms the sets 
    \[
    \overline{E}_i:=\left\{l\in\{1,2,\cdots,n\}:Y_{i,l}\neq e\right\}
    \]
    \[
    E_i:=\left\{l\in\{1,2,\cdots,n\}:Y_{i,l}=e\right\}
    \]
    If $\lvert \overline{E}_i\rvert<r_i n$ or $\lvert E_i \rvert<r_i n$, Bob-$i$ aborts the protocol.
    \item Bob-$i$ creates the following sets:
    \[
    S_{iZ_i}\sim \text{Unif}\left\{A\subset\overline{E}_i: \lvert A \rvert=r_i n\right\}
    \]
    \[
    S_{i\overline{Z}_i}\sim \text{Unif}\left\{A\subset E_i: \lvert A \rvert=r_i n\right\}
    \]
    
    \parbreak Bob-$i$ reveals $S_{i0}$ and $S_{i1}$  to Alice over the noiseless public channel.
    
    \item Alice randomly and independently chooses functions $\kappa_{i0}$ and $\kappa_{i1}$ from a family $\mathcal{K}$ of two-universal hash functions:
    \[
    \kappa_{i0}, \kappa_{i1}: \{0,1\}^{r_i n}\rightarrow \{0,1\}^{n(r_i-\lambda')}
    \]
    \[
    h_{i0},h_{i1} : \mathcal{R}_{ij} \times \mathcal{X}^n \to \{0,1\}^{s_i n}
    \]
    Alice finally sends the following information to Bob-$i$ over the noiseless public channel: 
    \[
    h_{i0}, h_{i1},\kappa_{i0}, \kappa_{i1}, m_{i0}\oplus\kappa_{i0}(\mathbf{X}|_{S_{i0}}), m_{i1}\oplus\kappa_{i1}(\mathbf{X}|_{S_{i1}})
    \]
    \item Bob-$i$ knows $\kappa_{iZ_i}$, $\mathbf{X}|_{S_{iZ_i}}$ one can decode $m_{iZ_i}$ so that: 
    \begin{itemize}
        \item $\hat{\mathbf{X}}|_{S_{iZ_i}}$ and $(\mathbf{Y}_1|_{S_{1Z_1}},\mathbf{Y}_2|_{S_{2Z_2}})$ are $\varepsilon$-conditional typical according to $W$, that is, $(\mathbf{Y}_1|_{S_{1Z_1}},\mathbf{Y}_2|_{S_{2Z_2}})\in \mathcal{T}^n_{W, \varepsilon}(\hat{\mathbf{X}}|_{S_{1Z_1},S_{2Z_2}})$;
        \item $h_{i}(R_{iZ_i}, \mathbf{X}|_{S_{iZ_i}}) = h_{i}(R_{iZ_i}, \hat{\mathbf{\mathbf{X}}}|_{S_{iZ_i}}), i\in\{1,2\}$;
        
        If there is more than one such sequence, or if no such sequence exists, Bob outputs an error.
        \end{itemize}
  \end{enumerate}
\end{protocol}

\begin{protocol}{OT over noisy DM-BEBC with colluding parties in the setup of Fig. \ref{fig: both}-$(b)$.\label{protocol: BC-Coll}}
       \textit{Parameters:} \begin{itemize}
    \item $\lambda\in (0,1)$ such that $r_i<p_{\overline{i}}\min \{p_i,1-p_i\}-\lambda, i \in \{1,2\}$, where $p_i$ is the erasure probability.
    \item $0<\lambda'<r_i, \lambda'\in \mathbb{Q}$
    \item $n(r_i-\lambda')\in\mathbb{N}$
    \item The rate of the protocol is $r_i-\lambda' = \frac{k_i}{n}$, and $k_i$ is the length of strings.
\end{itemize}
\textit{Goal.} Alice sends two strings $m_{i0},m_{i1}, i \in \{1,2\}$ to Bob-$i$. At the end of the protocol, Bob-$i$ gets $m_{iZ_i}, Z_i\in \{0,1\}$ while Alice gets nothing $e$.

\textit{The protocol:}
  \begin{enumerate}
    \item Alice transmits an $n$-tuple $\mathbf{X}= X^n$ of i.i.d. Bernoulli($\frac12$) bits over one of the noisy point-to-point BEC($p_i$) based on her choice. 

    \item Bob-$i$ receives the $n$-tuple $\mathbf{Y_i}= Y_i^n$ from BEC($p_i$). Bob-$i$ forms the sets 
    \[
    \overline{E}_i:=\left\{l\in\{1,2,\cdots,n\}:Y_{i,l}\neq e\right\}
    \]
    \[
    E_i:=\left\{l\in\{1,2,\cdots,n\}:Y_{i,l}=e\right\}
    \]
    If $\lvert \overline{E}_i\rvert<r_i n$ or $\lvert E_i \rvert<r_i n$, Bob-$i$ aborts the protocol.
    \item Bob-$i$ creates the following sets:
    \[
    S_{Z_i}\sim \text{Unif}\left\{A\subset\overline{E}: \lvert A \rvert=\frac{r_i}{p_{\overline{i}}-\lambda'} n\right\}
    \]
    \[
    S_{\overline{Z}_i}\sim \text{Unif}\left\{A\subset E_i: \lvert A \rvert=\frac{r_i}{p_{\overline{i}}-\lambda'} n\right\}
    \]
    If $p_i>\frac12$,
    \[
    S'\sim \text{Unif}\left\{A\subset E_i\setminus [S_{\overline{Z}_i}]: \lvert A \rvert=(p_i-\lambda-\frac{r_i}{p_{\overline{i}}-\lambda'}) n\right\}
    \]
    Else,
    
    \qquad$S'=\emptyset$
    
    \parbreak Bob-$i$ reveals $S_0, S_1$ and $S'$ to Alice over the noiseless public channel.
    
    \item Alice randomly and independently chooses functions $\kappa_{i0}$ and $\kappa_{i1}$ from a family $\mathcal{K}$ of two-universal hash functions:
    \[
    \kappa_{i0}, \kappa_{i1}: \{0,1\}^{r_i n}\rightarrow \{0,1\}^{n(r_i-\lambda')}
    \]
    Alice finally sends the following information to Bob-$i$ over the noiseless public channel: 
    \[
    \kappa_{i0}, \kappa_{i1}, m_{i0}\oplus\kappa_{i0}(\mathbf{X}|_{S_0}), m_{i1}\oplus\kappa_{i1}(\mathbf{X}|_{S_1})
    \]
    \end{enumerate}
\end{protocol}

\clearpage
\addtocounter{protocol}{-1}

\begin{protocol}{OT over noisy DM-BEBC with colluding parties in the setup
of Fig.~\ref{fig: both}-(b) (continued).}

\begin{enumerate}
\setcounter{enumi}{4}
    \item Bob-$i$ knows $\kappa_{iZ_i}$, $\mathbf{X}|_{S_{Z_i}}$ one can decode $m_{iZ_i}$ so that:
     \begin{itemize}
        \item $\hat{\mathbf{X}}|_{S_{Z_i}}$ and $\mathbf{Y}_i|_{S_{Z_i}}$ are $\varepsilon$-conditional typical according to $W$, that is, $\mathbf{Y}_i|_{S_{Z_i}}\in \mathcal{T}^n_{W, \varepsilon}(\hat{\mathbf{X}}|_{S_{Z_i}})$;
        \item $h_{i}(R_{iZ_i}, \mathbf{X}|_{S_{Z_i}}) = h_{i}(R_{iZ_i}, \hat{\mathbf{\mathbf{X}}}|_{S_{Z_i}}), i\in\{1,2\}$;
        
        If there is more than one such $\hat{\mathbf{X}}|_{S_{Z_i}}$ or none, Bob outputs an error.
        \end{itemize}
    \item If $S'\neq \emptyset$, then Alice transmits $\mathbf{X}|_{S'}$ over the noisy point-to-point BEC($p_{\overline{i}}$) to Bob-$\overline{i}$.
    \item Bob-$\overline{i}$ receives $\mathbf{Y_{\overline{i}}}$ from BEC($p_{\overline{i}}$). Bob-${\overline{i}}$ forms the sets 
    \[
    \overline{E_{\overline{i}}}:=\left\{l\in\{1,2,\cdots,n\}:Y_{{\overline{i}},l}\neq e\right\}
    \]
    \[
    E_{\overline{i}}:=\left\{l\in\{1,2,\cdots,n\}:Y_{{\overline{i}},l}=e\right\}
    \]
    If $\lvert \overline{E}_{\overline{i}}\rvert<r_{\overline{i}} n$ or $\lvert E_{\overline{i}} \rvert<r_{\overline{i}} n$, Bob-${\overline{i}}$ aborts the protocol.
    \item Bob-${\overline{i}}$ creates the following sets:
    \[
    S'_{Z_{\overline{i}}}\sim \text{Unif}\left\{A\subset\overline{E}_{\overline{i}}: \lvert A \rvert=r_{\overline{i}} n\right\}
    \]
    \[
    S'_{\overline{Z}_{\overline{i}}}\sim \text{Unif}\left\{A\subset E_{\overline{i}}: \lvert A \rvert=r_{\overline{i}} n\right\}
    \]
    \item Alice uses her previously chosen two-universal hash functions:
    \[
    \kappa_{\overline{i}0}, \kappa_{\overline{i}1}: \{0,1\}^{r_{\overline{i}} n}\rightarrow \{0,1\}^{n(r_{\overline{i}}-\lambda')}
    \]
    Alice sends the following information to Bob-$\overline{i}$ over the noiseless public channel:
    \[
    \kappa_{\overline{i}0}, \kappa_{\overline{i}1}, m_{\overline{i}0}\oplus\kappa_{\overline{i}0}(\mathbf{X}|_{S'}|_{S'_0}), m_{\overline{i}1}\oplus\kappa_{\overline{i}1}(\mathbf{X}|_{S'}|_{S'_1})
    \]
    \item Bob-$\overline{i}$ knows $\kappa_{\overline{i}Z_{\overline{i}}}$, $(\mathbf{X}|_{S'})|_{S'_{Z_{\overline{i}}}}$ one can decode $m_{\overline{i}Z_{\overline{i}}}$ so that:
     \begin{itemize}
        \item $(\hat{\mathbf{X}}|_{S'})|_{S'_{Z_{\overline{i}}}}$ and $(\mathbf{Y}_{\overline{i}}|_{S'})|_{S'_{Z_{\overline{i}}}}$ are $\varepsilon$-conditional typical according to $W$, that is, $(\mathbf{Y}_{\overline{i}}|_{S'})|_{S'_{Z_{\overline{i}}}}\in \mathcal{T}^n_{W, \varepsilon}((\hat{\mathbf{X}}|_{S'})|_{S'_{Z_{\overline{i}}}})$;
        \item $h_{\overline{i}}(R_{\overline{i}Z_{\overline{i}}}, (\mathbf{X}|_{S'})|_{S'_{Z_{\overline{i}}}}) = h_{\overline{i}}(R_{\overline{i}Z_{\overline{i}}}, (\hat{\mathbf{\mathbf{X}}}|_{S'})|_{S'_{Z_{\overline{i}}}}), i\in\{1,2\}$;
        
        If there is more than one such $\hat{(\mathbf{X}}|_{S'})|_{S'_{Z_{\overline{i}}}}$ or none, Bob-${\overline{i}}$ outputs an error.
        \end{itemize}
  \end{enumerate}
  \vspace{0.5pt}
\end{protocol}

\begin{theorem}\label{thm:BC-coll}
    In a general $\binom{2}{1}\emph{-OT}^{k_i}$ setup over a DM-BC $\mathcal{W}:\mathcal{X}\rightarrow \mathcal{Y}_1\mathcal{Y}_2:P_{Y_1Y_2|X}(y_1,y_2|x)$ (Fig. \ref{fig: both}-(a)) with honest-but-curious colluding  parties, if an OT rate pair is achievable, then it belongs to the set of all rate pairs $(R_1,R_2)\in\mathbb{R}^2_{+}$ that satisfy:
    
    \begin{align*}
        R_1 &\leq \min\Big\{\max_{P_{X}}I(X;Y_1), \max_{P_{X}}I(X;Y_1|Y_2),\max_{P_{X}}H(X|Y_1,Y_2)\Big\},\\
        R_2 &\leq \min\Big\{\max_{P_{X}}I(X;Y_2), \max_{P_{X}}I(X;Y_2|Y_1),\max_{P_{X}}H(X|Y_1,Y_2)\Big\},\\
        R_1+R_2 &\leq \min\Big\{\max_{P_{X}}I(X;Y_1,Y_2), \max_{P_{X}}H(X|Y_1,Y_2)\Big\},
    \end{align*}
\end{theorem}

\begin{proof}
    The proof is presented in \ref{app1:proof of thm:BC-coll}.
\end{proof}
\begin{proposition}\label{prop1}
    The rate region $\mathcal{R}$ of independent pairs of $\binom{2}{1}\emph{-OT}^{k_i}$ protocols over the DM-BC  (Fig. \ref{fig: both}-(b)): $\mathcal{W}:\mathcal{X}\rightarrow \mathcal{Y}_1\mathcal{Y}_2:P_{Y_1|X}(y_1|x)P_{Y_2|X}(y_2|x)$ with colluding parties is such that $\mathcal{R}\subseteq\mathcal{R}_{\text{outer}}$:
    \begin{align*}
    \mathcal{R}_{\text{outer}} = \left\{\ 
    \begin{aligned}
    (R_1,R_2)\in\mathbb{R}^2_{+}:R_1 &\leq p_2 \min\{p_1, 1 - p_1\}, \\
    R_2 &\leq p_1 \min\{p_2, 1 - p_2\}, \\
    R_1 + R_2 &\leq \min \{p_1p_2, 1-p_1p_2\}
    \end{aligned}
    \right\}
    \end{align*}
    Note that, in the above rate region, the sum rate is redundant except if one of the sub-BECs is a completely erasure ($p_1 = 1$ or $p_2 = 1$). Also, the term $1-p_i, i\in\{1,2\}$ is trivial but only in the setup of Fig. \ref{fig: both}-(b):
    \begin{align*}%
    R_1 &\leq  p_2.\min\{p_1,1-p_1\},\\
    R_2 & \leq p_1.\min\{p_2,1-p_2\},\\
    R_1+R_2 & \leq \min\{p_1p_2, 1-p_1p_2\}.
    \end{align*}
\end{proposition}

Now, we introduce two OT protocols. The first protocol supports pairwise OT between Alice and each of the two non-colluding receivers, allowing her to engage in independent OT sessions with each party. The second protocol remains secure even if the receivers collude. It adopts a sequential strategy: Alice first conducts an OT protocol with one receiver, and upon completion, initiates a separate OT session with the other.

\parbreak Let $s_i$ and $r_i$, for $i \in \{1,2\}$, be four parameters, and let 
\[
h_{ij} : \mathcal{R}_{ij} \times \mathcal{X}^n \to \{0,1\}^{s_i n}, 
\qquad 
\kappa_{ij} : \mathcal{T}_{ij} \times \mathcal{X}^n \to \{0,1\}^{r_i n}, 
\]
for $i \in \{1,2\}, j \in \{0,1\}$, be two-universal hash functions, and $\mathcal{R}_{ij}\subset R = (R^{(1)},R^{(2)}), \mathcal{T}_{ij}\subset T= (T^{(1)}=({\mathcal{T}_{10},\mathcal{T}_{11}}),T^{(2)}=({\mathcal{T}_{20},\mathcal{T}_{21}}))$ are seeds chosen uniformly at random from the randomness sets $R$ and $T$, respectively.  
The $\binom{2}{1}$-OT$^{k_1,k_2}$ rate pair is defined as 
\[
(r_1, r_2) = \left(\lim_{n\to\infty}\frac{k_1}{n}, \lim_{n\to\infty}\frac{k_2}{n}\right).
\]  
In this setting, the strings are encrypted using the hash functions 
$\kappa_{ij}$, while the functions $h_{ij}$ are employed for privacy amplification.

\begin{lemma}\label{lem: cri. nColl-BC}
    Protocol 1 is correct (reliable) and secure against non-colluding honest-but-curious players even without privacy amplification, but is not secure against colluding players. 
\end{lemma}
\begin{proof}
    The proof is presented in \ref{app: proof of lem: cri. nColl-BC}.
\end{proof}
\parbreak

A rate pair $(R_1,R_2)$ is achievable if there exists
a sequence of $\binom{2}{1}$-OT$^{k_1,k_2}$ protocols satisfying \eqref{goals: BC-nColl-1}-\eqref{goals: BC-nColl-3}
such that\linebreak $\lim_{n\to\infty}(\frac{k_1}{n},\frac{k_2}{n})= (R_1, R_2)$.
\begin{theorem}\label{thm: OT capacity-non-col}
    In a general setup over a DM-BC consisting of two independent BECs depicted in Fig.~\ref{fig: both}-(b), the OT capacity region of the broadcast channel $\mathcal{W}:\mathcal{X}\rightarrow \mathcal{Y}_1\mathcal{Y}_2:P_{Y_1Y_2|X}(y_1,y_2|x) = P_{Y_1|X}(y_1|x)P_{Y_2|X}(y_2|x)$ with non-colluding parties is:
    \begin{align}%
    \mathcal{C}_{\text{OT}}& =
    \left\{ 
    \begin{array}{rl}
    (R_1,R_2)\in\mathbb{R}^2_{+} \,:\;
	R_1 &\leq  \max_{P_{X}} I(X;Y_1)\\
    R_2 & \leq \max_{P_{X}} I(X;Y_2)\\
	R_1+R_2 & \leq \max_{P_{X}} I(X;Y_1,Y_2)
    \end{array}
    \right\},
    \end{align}
\end{theorem}
\begin{proof}
    The proof is presented in \ref{app: thm: OT capacity-non-col}.
\end{proof}
Note that, for the DM-BC in Fig. \ref{fig: both}-$
(b)$ the capacity region is equal to:
    \begin{align}%
    \mathcal{C}_{\text{OT}}& =
    \left\{ 
    \begin{array}{rl}
    (R_1,R_2)\in\mathbb{R}^2_{+} \,:\;
	R_1 &\leq  1-p_1\\
    R_2 & \leq 1-p_2\\
	R_1+R_2 & \leq 1-p_1p_2
    \end{array}
    \right\},
    \end{align} 

\parbreak We present a slightly modified protocol based on \cite{Mishra2} that is secure against colluding parties.

\begin{lemma}\label{lem: cri. Coll-BC}
    Protocol 2 is correct and secure against colluding honest-but-curious players.
\end{lemma}
\begin{proof}
    The proof is presented in \ref{app: proof of lem: cri. Coll-BC}.
\end{proof}
\begin{proposition}\cite{mishra}\label{prop2}
    In a general setup over a DM-BC consisting of two independent BECs in the setup of Fig. \ref{fig: both}-$(b)$, the rate region $\mathcal{R}$ of Protocol 2, for colluding honest-but-curious users is such that $\mathcal{R}_{\text{inner}}\subseteq\mathcal{R}\subseteq\mathcal{R}_{\text{outer}}$, where ($(R_1, R_2) \in \mathbb{R}_+^2$):
    
\begin{align*}
\mathcal{R}_{\text{inner}} = \left\{\ 
\begin{aligned}
R_1 &\leq p_2 \min\{p_1, 1 - p_1\}, \\
R_2 &\leq p_1 \min\{p_2, 1 - p_2\}, \\
R_1 + R_2 &\leq\ p_2 \min\{p_1, 1 - p_1\} \\
& \,\,\, + p_1 \min\{p_2, 1 - p_2\} \\
& \,\,\,- \min\{p_1, 1 - p_1\} \cdot \min\{p_2, 1 - p_2\}
\end{aligned}
\right\},
\end{align*}
and $\mathcal{R}_{\text{outer}}$ is presented in Proposition \ref{prop1}.
\end{proposition}
\begin{proof}\label{thm: lower: coll}
    The proof is presented in \ref{app: thm: lower: coll}.
\end{proof}
\section{Conclusion}\label{Sec: conc}
We study the problem of OT over a noisy DM-BC with one sender and two receivers, considering both colluding and non-colluding honest-but-curious parties. General upper bounds on the OT capacity are established for both cases. The channel model is then reduced to a noisy DM-BC consisting of two independent sub-BECs with different noise levels. For the case of non-colluding parties, the lower and upper bounds coincide, whereas for colluding parties, they do not generally coincide. However, the case in which the players are also allowed to cheat (i.e., malicious behavior or active attacks) remains an open problem.

\section*{Acknowledgments}
\begin{sloppy}
   
The authors acknowledge financial support from the German Federal Ministry of Research, Technology and Space (BMFTR) under the ``Souverän. Digital. Vernetzt.'' programme through the project 6G-life (grants 16KIS2414 and 16KIS2415), as well as through the Quantum Programme projects QuaPhySI (grants 16KISQ1598K and 16KIS2234), QR.N (grants 16KIS2195 and 16KIS2196), Q-STARS (grants 16KIS2601 and 16KIS2602), QUIET (grants 16KISQ093 and 16KISQ0170), QD-CamNetz (grants 16KISQ077 and 16KISQ169), and Q-TREX (grants 16KISR027K and 16KISR038).
 
\end{sloppy}
H.~Boche acknowledges funding from the German Research Foundation (DFG) under Germany's Excellence Strategy—EXC 2050/2, Project ID 390696704, Cluster of Excellence ``Centre for Tactile Internet with Human-in-the-Loop'' (CeTI), Technische Universität Dresden.

Parts of the results presented in this paper were previously reported in the conference paper~\cite{Hadi-ICC26}.

\clearpage
\section*{Appendix}
\appendix
\renewcommand{\thesection}{Appendix~\Alph{section}}
\renewcommand{\theHsection}{appendix.\Alph{section}}

\section{Proof of Corollary~\ref{cor}}
\label{app: cor}

This is intuitively reasonable. Consider the point-to-point OT in the presence of a passive wiretapper who can collude with Alice or Bob. If an OT protocol provides perfect secrecy, then it should also provide perfect secrecy in the presence of a further wiretapper since the wiretapper has access only to the noiseless public channel and does not interact directly with either party. In other words, the wiretapper gains nothing more than the other parties already have. Consider the following setup:

\parbreak Let $Z$ denote the receiver's choice bit (or choice variable) and let $M_{\overline{Z}}$ denote the sender's secret that should remain hidden.   
Let $U$ and $V_B$ denote Alice's and Bob's final views, respectively.  
The random variable $U$ includes the total public transmission $\mathbf{C}$, and Alice's randomness $R_A$, and the random variable $V_B$ includes the total public transmission $\mathbf{C}$, and Bob's randomness $R_B$.   
Formally, $U = (\mathbf{C},R_A)$, $V_B = (\mathbf{C},R_B)$. 

\parbreak Let $V_W$ denote the wiretapper's knowledge when colluding with Alice (or Bob).  
By the collusion model, $V_W$ is some (possibly randomized or deterministic) function of $U$ (or $V_B$) and possibly $\mathbf{C}$.  
Formally we write in the case where the wiretapper and Alice collude
\[
  V_W = g(U),
\]
for some function $g$.
When a perfect OT is possible, the protocol guarantees
\begin{equation}
I(Z;U)=0
\quad \text{and} \quad
I(M_{\overline{Z}};V_B,V_W)=0,
\end{equation}
i.e.\ the honest party's full view reveals zero mutual information about the protected secrets.

Under the above collusion model, the wiretapper learns no information about the secrets:
\[
I(Z;V_W)=0 \quad\text{and}\quad I(M_{\overline{Z}};V_W)=0.
\]

Since $V_W$ is obtained as a (possibly randomized) function of $U$, we may apply the data-processing inequality:
\[
I(Z;V_W) \le I(Z;U).
\]
By the perfect secrecy assumption, $I(Z;U)=0$. 
Thus, $I(Z;V_W)=0$.

On the other hand:
\[
I(M_{\overline{Z}};V_B,V_W) = I(M_{\overline{Z}};V_W) + I(M_{\overline{Z}};V_B|V_W) = 0.
\]
Thus, $I(M_{\overline{Z} };V_W) = 0$.

Therefore, the colluding wiretapper, whose knowledge is nothing more than data computable from the other party's views, cannot increase its information about the secrets beyond zero. 

\section{Proof of Theorem \ref{thm:BC-coll}\label{app1:proof of thm:BC-coll}}

\parbreak Consider Eq. \eqref{goals: BC-Coll-2}:
\begin{align}
    \lim_{n\rightarrow\infty}I (  M_{1\overline{Z}_1},M_{2\overline{Z}_2};V_1,V_2) & \,= I (  M_{1\overline{Z}_1};V_1,V_2) + I (  M_{2\overline{Z}_2};V_1,V_2|M_{1\overline{Z}_1})\notag\\
    & \,= I (  M_{1\overline{Z}_1};V_1,V_2) + I (  M_{2\overline{Z}_2};M_{1\overline{Z}_1}, V_1,V_2)\notag\\
    & \overset{(a)}{=} \sum_{i=[1:2]}\lim_{n\rightarrow\infty}I (  M_{i\overline{Z}_i};V_i,V_{\overline{i}})\\
    & \,= 0,\notag
\end{align}
where $(a)$ follows from the independence of $(M_{1\overline{Z}_1},V_1)$ from $M_{2\overline{Z}_2}$ given $V_2$.

\parbreak This means that both unselected messages should be hidden from both receivers separately. Consider the smooth version of $\lim_{n\rightarrow\infty}I (  M_{i\overline{Z}_i};V_i,V_{\overline{i}}) = 0$ as:

\begin{equation}\label{smoothed1}
    I (  M_{i\overline{Z}_i};V_i,V_{\overline{i}})_{i \in \{1,2\}} = I (  M_{i\overline{z}_i};V_i,V_{\overline{i}}|Z_i = z_i)_{i \in \{1,2\}} = o(n).
\end{equation}


Now consider Eq. \eqref{goals: BC-Coll-3}:
\begin{align}\label{smoothed2}
    \lim_{n\rightarrow\infty}I (Z_i; U, V_{\overline{i}})_{i\in\{1,2\}} & = \lim_{n\rightarrow\infty}I (Z_i; M_{i0}, M_{i1}, R_{A}, X^n , \mathbf{C}_i, V_{\overline{i}})_{i\in\{1,2\}}\notag\\
    & = \lim_{n\rightarrow\infty}I (Z_i=z_i; M_{i\overline{z}_i}, R_{A}, X^n , \mathbf{C}_i, V_{\overline{i}}|Z_i=z_i)_{i\in\{1,2\}}\\
    & = 0.\notag
\end{align}
By applying Lemma \ref{lemma: Rudolf} to the smooth version of Eq. \eqref{smoothed2} without considering the random variable representing local randomness, we have:

\begin{equation}\label{smoothed: applied lemma}
    H(M_{i\overline{z}_i}| X^n, \mathbf{C}_i, V_{\overline{i}},Z_i = \overline{z}_i)_{i\in\{1,2\}} - H(M_{i\overline{z}_i}| X^n, \mathbf{C}_i, V_{\overline{i}}, Z_i = z_i)_{i\in\{1,2\}} = o(n).
\end{equation}
Since $H(M_{i\overline{z}_i}|Z_i=z_i, V_{\overline{i}})_{i\in\{1,2\}} = H(M_{i\overline{z}_i}|Z_i=\overline{z}_i, V_{\overline{i}})_{i\in\{1,2\}}$, then Eq. \eqref{smoothed: applied lemma} can be written as follows:

\begin{equation}
    I(M_{i\overline{z}_i};\mathbf{C}_i, V_{\overline{i}}|Z_i = \overline{z}_i)_{i\in\{1,2\}} = I(M_{i\overline{z}_i};\mathbf{C}_i, V_{\overline{i}}|Z_i = z_i)_{i\in\{1,2\}} + o(n).
\end{equation}
Then, from Eq. \eqref{smoothed1}, we can conclude: 

\begin{equation}\label{final: BC-coll}
    I(M_{i{z}_i};\mathbf{C}_i, V_{\overline{i}}|Z_i = {z}_i)_{i\in\{1,2\}}  = I(M_{i\overline{z}_i};\mathbf{C}_i, V_{\overline{i}}|Z_i = \overline{z}_i)_{i\in\{1,2\}}       = o(n).
\end{equation}
We can show that:
\begin{align}\label{final: BC-coll-2}
    I(M_{i{z}_i};\mathbf{C}_i, V_{\overline{i}})_{i\in\{1,2\}} & = I(M_{i\overline{z}_i};\mathbf{C}_i, V_{\overline{i}})_{i\in\{1,2\}} \notag\\
    & = \sum_{z_i\in\{0,1\}} \text{Pr}[{Z_i = z_i}]\,.\, I(M_{i\overline{z}_i};\mathbf{C}_i, V_{\overline{i}}|Z_i = \overline{z}_i)_{i\in\{1,2\}}\notag\\
    & = o(n).
\end{align}
Also, another form of \eqref{final: BC-coll} could be:
\begin{align}
    I(M_{i{z}_i};\mathbf{C}_i, V_{\overline{i}}|Z_i = {z}_i)_{i\in\{1,2\}} & = I(M_{i{z}_i};\mathbf{C}_i, Z_{\overline{i}}, R_{B_i}, Y_{\overline{i}}^n|Z_i = {z}_i)_{i\in\{1,2\}} \notag\\
    & \leq I(M_{i{z}_i};\mathbf{C}_i, Y_{\overline{i}}^n|Z_i = {z}_i)_{i\in\{1,2\}} \notag\\
    & = o(n).
\end{align}
Then, by the same reasoning as \eqref{final: BC-coll-2}, we have:
\begin{equation}\label{final: BC-coll-3}
     I(M_{i{z}_i};\mathbf{C}_i, Y_{\overline{i}}^n)_{i\in\{1,2\}}  = o(n).
\end{equation}

First, we want to prove for every OT protocol that follows the general two party protocol described before, we have,
\begin{align}
    \label{eq. SKA-main1}
    k_i &= H(M_{iz_i}|Z_i=\overline{z}_i) \notag\\
    & = H(M_{i\overline{z}_i}|Z_i=\overline{z}_i)\notag\\
    & \leq H(M_{iz_i})\notag\\
    & \leq I(X; Y_i | Y_{\overline{i}})+ I(M_{iz_i}; \mathbf{C}_i, Y_{\overline{i}}) + H(M_{iz_i} | \hat{M}_{iz_i}).
\end{align}
In particular, if we consider $Y_{\overline{i}}$ is constant random variable for $X$ and $Y_i, i\in\{1,2\}$ then:
\begin{equation}\label{eq. SKA-main2}
    H(M_{iz_i}) \leq I(X; Y_i) + I(M_{iz_i}; \mathbf{C}_i) + H(M_{iz_i} | \hat{M}_{iz_i}).
\end{equation}
We have,
\begin{equation}\label{eq. SKA1}
    H(M_{iz_i}) = I(M_{iz_i}; \mathbf{C}_i ,Y_{\overline{i}}) + H(M_{iz_i} | \mathbf{C}_i , Y_{\overline{i}}), 
\end{equation}
where $H(M_{iz_i} | \mathbf{C}_i, Y_{\overline{i}})$ can be bounded from above as follows:
\begin{align}\label{eq. SKA2}
H(M_{iz_i} | \mathbf{C}_i, Y_{\overline{i}}) &= H(M_{iz_i},X | \mathbf{C}_i, Y_{\overline{i}}) - H(X | \mathbf{C}_i, M_{iz_i},Y_{\overline{i}})\notag \\
&= H(X | \mathbf{C}_i, Y_{\overline{i}}) + H(M_{iz_i} | \mathbf{C}_i, X,Y_{\overline{i}}) - H(X | \mathbf{C}_i, M_{iz_i},Y_{\overline{i}}) \notag\\
&\stackrel{(a)}{\leq} H(X | \mathbf{C}_i, Y_{\overline{i}}) - H(X  | \mathbf{C}_i, M_{iz_i},Y_i,Y_{\overline{i}})\notag\\
& = H(X | \mathbf{C}_i, Y_{\overline{i}}) - H(X,M_{iz_i} | \mathbf{C}_i, Y_i,Y_{\overline{i}}) + H(M_{iz_i} | \mathbf{C}_i, Y_i,Y_{\overline{i}})\notag \\
&\stackrel{(b)}{=} H(X | \mathbf{C}_i, Y_{\overline{i}}) - H(X | \mathbf{C}_i, Y_i,Y_{\overline{i}}) + H(M_{iz_i} | \mathbf{C}_i, Y_i,Y_{\overline{i}}) \notag\\
&\stackrel{(c)}{\leq} I(X; Y_i | \mathbf{C}_i, Y_{\overline{i}}) + H(M_{iz_i} | \hat{M}_{iz_i}), 
\end{align}
for $i\in\{1,2\}$,  where $(a)$ and $(b)$ follow immediately from the fact that\linebreak $H(M_{iz_i}|\mathbf{C}_i, X) = 0$ and the fact that conditioning does not increase entropy. $(c)$ follows from the fact that $H(\hat{M}_{iz_i}|\mathbf{C}_i, Y_i) = 0$ that implies $H(M_{iz_i} | \mathbf{C}_i ,Y_i,Y_{\overline{i}}) = H(M_{iz_i} | \hat{M}_{iz_i},\mathbf{C}_i, Y_i,Y_{\overline{i}}) \leq H(M_{iz_i}|\hat{M}_{iz_i})$. 
\begin{align}
    H(M_{iz_i} | \mathbf{C}_i, Y_i,Y_{\overline{i}}) & = H(S | \mathbf{C}_i, \hat{M}_{iz_i}, Y_i,Y_{\overline{i}})\notag\\
    & \leq H(M_{iz_i} | \hat{M}_{iz_i}).
\end{align}

Now we show that using the public noiseless channel cannot increase the mutual information shared by Alice and Bob-$i$ when given Bob-$\overline{i}$’s total information consisting of $\mathbf{C}_i$ and $Y_{\overline{i}}$. Without loss of generality, assume that $n$ is odd, i.e., that the last public message $\mathbf{C}_i$ is sent by Alice and thus $H(\mathbf{C}_i | C_i^{n-1}, X) = 0, i\in\{1,2\}$. The proof for even $n$ is analogous.
\begin{align}
I(X; Y_i | \mathbf{C}_i, Y_{\overline{i}}) &= H(Y_i | \mathbf{C}_i, Y_{\overline{i}}) - H(Y_i | \mathbf{C}_i, X,Y_{\overline{i}}) \notag\\
& \stackrel{(a)}{=} H(Y_i | \mathbf{C}_i ,Y_{\overline{i}}) - H(Y_i | C_i^{n-1} ,X,Y_{\overline{i}})\notag \\
&\leq H(Y_i | C_i^{n-1} ,Y_{\overline{i}}) - H(Y_i | C_i^{n-1} ,X,Y_{\overline{i}})\notag \\\label{eq. SKA3'}
& = I(X; Y_i | C_i^{n-1} ,Y_{\overline{i}})\\\label{eq. SKA3}
& \stackrel{(b)}{\leq} I(X;Y_i|Y_{\overline{i}}),
\end{align}
where $(a)$ follows from the fact that $H(C_{i,k}|C_i^{k-1}, X) = 0$ for odd $k$ ($C_{i,k}$ is a function of $U, R_A, \text{and}\,\{C_j, j<k\}$) and $H(C_{i,k}|C_i^{k-1}, Y) = 0$ for even $k$ ($C_{i,k}$ is a function of $V_i, R_{B_i}, \text{and}\,\{C_j, j<k\}$), and $(b)$ is due to repeating the above argument $n$ times. Putting \eqref{eq. SKA1}, \eqref{eq. SKA2}, and \eqref{eq. SKA3} together, proves \eqref{eq. SKA-main1}.

\parbreak Now, consider the last term of \eqref{eq. SKA-main1}. Applying Fano's lemma to \eqref{goals: BC-Coll-1}, we have:
\begin{equation}\label{eq. SKA4}
    H(M_{iz_i} | \hat{M}_{iz_i})\leq h(\varepsilon) + \log (|\mathcal{M}_{iz_i}|-1).
\end{equation}

Let $|\mathcal{M}_{iz_i}|$ denote the number of distinct values that the random variable $M_{iz_i}$ can take with nonzero probability. Note that $H(M_{iz_i} \mid \hat{M}_{iz_i}) \to 0$ as $\varepsilon \to 0$. If we require that $\Pr[M_{iz_i} \neq \hat{M}_{iz_i}] = 0$, and $I(M_{iz_i}; C_i^n) = 0$, it may seem intuitive—though not immediately obvious—that $I(X; Y_i)$ serves as an upper bound on $H(M_{iz_i})$. Similarly, it appears intuitive that $H(\hat{M}_{iz_i}) \leq I(X; Y_i \mid Y_{\overline{i}}) = I(X,Y_{\overline{i}}; Y_i,Y_{\overline{i}}) - H(Y_{\overline{i}})$ since even if Alice and Bob-$i$ were to learn $Y_{\overline{i}}$, the mutual information they could share remains bounded by the information they can hold in common, minus the uncertainty due to $Y_{\overline{i}}$.

\parbreak Now combining \eqref{eq. SKA4}, \eqref{eq. SKA-main1}, and \eqref{eq. SKA-main2}, we have:
\begin{align}
    k_i \leq H(M_{iz_i}) \leq \min \Big\{ I(X; Y_i | Y_{\overline{i}}), I(X; Y_i)\Big\} &+ I(M_{i{z}_i};\mathbf{C}_i, Y_{\overline{i}}^n)\notag\\
    & + h(\varepsilon) + \log (|\mathcal{M}_{iz_i}|-1),
\end{align}
Now, we can calculate an upper bound on OT capacity:
\begin{align}
    R_i = \lim_{n\to\infty} \frac{k_i}{n} &\leq \min \Big\{ I(X; Y_i | Y_{\overline{i}}), I(X; Y_i)\Big\}\notag\\
    & \leq \min \Big\{ \max_{P_X} I(X; Y_i | Y_{\overline{i}}), \max_{P_X} I(X; Y_i)\Big\},
\end{align}
for $i\in\{1,2\}$. The first inequality follows from \eqref{final: BC-coll-3} and the fact that $H(M_{iz_i} \mid \hat{M}_{iz_i}) \to 0$ as $\varepsilon \to 0$, and the last inequality follows from the following argument:

In proving the upper bound, our goal is to show:
\begin{equation}\label{(19)}
R_i = \frac{1}{n} H(M_{iZ_i}) \leq \max_{P_X} I(X;Y_i|Y_{\overline{i}}) + \frac{1}{n} H(M_{iZ_i}|\hat{M}_{iZ_i}) 
+ \frac{1}{n} I(M_{iZ_i}; C_i^n ,Y_{\overline{i}}^n). 
\end{equation}

According to \eqref{final: BC-coll-3}, \eqref{eq. SKA4} and the definition of secrecy capacity the last two terms go to zero as $n \to \infty$.  
This means that  $H(M_{iZ_i})/n \leq \max_{P_X} I(X;Y_i|Y_{\overline{i}})$,and by choosing $Y_{\overline{i}}$ as a constant random variable would be \[
H(M_{iZ_i})/n \leq \max_{P_X} I(X;Y_i).
\]

\parbreak In order to prove \eqref{(19)}, note that $H(M_{iZ_i}) \leq I(\mathbf{X}; \mathbf{Y}_i | \mathbf{C}_i ,\mathbf{Y}_{\overline{i}}) +  H(M_{iZ_i}|\hat{M}_{iZ_i})
+ I(M_{iZ_i}; \mathbf{C}_i ,\mathbf{Y}_{\overline{i}})$
can be obtained in a manner entirely analogous to the derivations of \eqref{eq. SKA1} and \eqref{eq. SKA2}. What is left is to demonstrate that
\begin{equation}\label{main}
I(\mathbf{X}; \mathbf{Y}_i | \mathbf{C}_i, \mathbf{Y}_{\overline{i}}) \leq n \cdot \max_{P_X} I(X; Y_i | Y_{\overline{i}}). 
\end{equation}
Let $l\in\{1,\ldots, n\}$.  
Alice’s choice of $P_{X_l}$ depends only on $C_i^{l-1}$ and $X^{l-1}$, but not further on $Y_i^{l-1}$ and $Y_{\overline{i}}^{l-1}$, i.e.,
\begin{equation}\label{(21)}
H(X_l | C_i^{l-1}, X^{l-1}, Y_i^{l-1}, Y_{\overline{i}}^{l-1}) = H(X_l | C_i^{l-1}, X^{l-1}, Y_{\overline{i}}^{l-1}) 
= H(X_l | C_i^{l-1}, X^{l-1}). 
\end{equation}

Likewise, the $l$-th broadcast channel output $[Y_{i,l}, Y_{\overline{i},l}]$ is influenced by $C_i^{l-1}$, $X^{l-1}$, $Y_i^{l-1}$, and $Y_{\overline{i}}^{l-1}$ solely through its dependence on $X_i$, which can be expressed as
\begin{equation}\label{(22)}
H(Y_{i,l}, Y_{\overline{i},l} | C_i^{l-1}, X^{l-1}, Y_i^{l-1}, Y_{\overline{i}}^{l-1}) 
= H(Y_{i,l}, Y_{\overline{i},l} | X_l). 
\end{equation}

One similarly has
\begin{equation}\label{(23)}
H(Y_{\overline{i},l} | C_i^{l-1}, X^{l-1}, Y_i^{l-1}, Y_{\overline{i}}^{l-1})  = H(Y_{\overline{i},l} | C_i^{l-1}, X^{l-1}, Y_{\overline{i}}^{l-1}) 
= H(Y_{\overline{i},l} | X_l), 
\end{equation}

when only the $Y_{\overline{i}}$-output of the channel is considered.  

\parbreak We now derive two equalities that will be used later.  
First, using \eqref{(22)} and \eqref{(23)} one obtains
\begin{align}\label{(24)}
H(Y_{i,l} | C_i^{l-1}, X^{l-1}, Y_i^{l-1}, Y_{\overline{i}}^l) 
& = H(Y_{i,l}, Y_{\overline{i},l} | C_i^{l-1}, X^{l-1}, Y_i^{l-1}, Y_{\overline{i}}^{l-1}) \notag \\
&\quad - H(Y_{\overline{i},l} | C_i^{l-1}, X^{l-1}, Y_i^{l-1}, Y_{\overline{i}}^{l-1}) \notag \\
& = H(Y_{i,l}, Y_{\overline{i},l} | Y_{\overline{i},l}) - H(Y_{\overline{i},l}| X_l) \notag \\
& = H(Y_{i,l} | X_l, Y_{\overline{i},l}).
\end{align}

Second, expanding the following conditional entropy in two different ways,
\begin{align}\label{mid1}
H(X_l, Y_i^{l-1}, Y_{\overline{i},l} | C_i^{l-1}, X^{l-1}, Y_{\overline{i}}^{l-1})
& \,= H(X_l, Y_{\overline{i},l} | C_i^{l-1}, X^{l-1}, Y_{\overline{i}}^{l-1}) \notag\\
& \quad+ H(Y_i^{l-1} |C_i^{l-1}, X^{l}, Y_{\overline{i}}^l)\\
& \,= H(X_l, Y_{\overline{i},l} | C_i^{l-1}, X^{l-1}, Y_i^{l-1}, Y_{\overline{i}}^{l-1})\notag\\
& \quad + H(Y_i^{l-1} | C_i^{l-1}, X^{l-1}, Y_{\overline{i}}^{l-1})\notag\\
& \,= H(X_l | C_i^{l-1}, X^{l-1}, Y_i^{l-1}, Y_{\overline{i}}^{l-1})\notag\\
& \quad+ H(Y_{\overline{i},l} | C_i^{l-1}, X^{l}, Y_i^{l-1}, Y_{\overline{i}}^{l-1})\notag\\
&\quad + H(Y_i^{l-1} | C_i^{l-1}, X^{l-1}, Y_{\overline{i}}^{l-1})\notag\\
& \overset{(a)}{=} H(X_l | C_i^{l-1}, X^{l-1}, Y_{\overline{i}}^{l-1})\notag\\
& \quad+ H(Y_{\overline{i},l} | C_i^{l-1}, X^{l}, Y_{\overline{i}}^{l-1})\notag\\
&\quad + H(Y_i^{l-1} | C_i^{l-1}, X^{l-1}, Y_{\overline{i}}^{l-1})\notag\\\label{mid2}
& \,= H(X_l, Y_{\overline{i},l}|C_i^{l-1}, X^{l-1}, Y_{\overline{i}}^{l-1})\notag\\
&\quad + H(Y_i^{l-1} | C_i^{l-1}, X^{l-1}, Y_{\overline{i}}^{l-1}),
\end{align}
where $(a)$ in due to \eqref{(21)} and \eqref{(23)}. Now, by a simple comparison between \eqref{mid1} and \eqref{mid2}, we can find:
\begin{equation}\label{(25)}
     H(Y_i^{l-1} | C_i^{l-1}, X^{l}, Y_{\overline{i}}^{l}) = H(Y_i^{l-1} | C_i^{l-1}, X^{l-1}, Y_{\overline{i}}^{l-1}), 
\end{equation}
for $i \in\{1,2\}$. Repeating the argument similar to \eqref{eq. SKA3'} for every public message of step $l$, we have:
\begin{align}\label{(26')}
I(X^{l}; Y_i^{l} | C_i^{l}, Y_{\overline{i}}^{l}) &\, \leq I(X^{l}; Y_i^{l} | C_i^{l-1}, Y_{\overline{i}}^{l})\notag\\
& \,= H(Y_i^{l} | C_i^{l-1}, Y_{\overline{i}}^{l}) - H(Y_i^{l} | C_i^{l-1}, X^{l}, Y_{\overline{i}}^{l}) \notag \\
& \,= H(Y_i^{l-1} | C_i^{l-1}, Y_{\overline{i}}^{l}) 
  + H(Y_{i,l} | C_i^{l-1}, Y_i^{l-1}, Y_{\overline{i}}^{l}) \notag \\
&\quad - H(Y_i^{l-1} | C_i^{l-1}, X^{l}, Y_{\overline{i}}^{l})
  - H(Y_{i,l} | C_i^{l-1} X^{l}, Y_i^{l-1}, Y_{\overline{i}}^{l}) \notag \\
&\overset{(a)}{\leq} H(Y_i^{l-1} | C_i^{l-1}, Y_{\overline{i}}^{l-1}) + H(Y_{i,l} | Y_{\overline{i},l}) \notag \\
&\quad - H(Y_i^{l-1} | C_i^{l-1}, X^{l-1}, Y_{\overline{i}}^{l-1})
  - H(Y_{i,l} | X_l, Y_{\overline{i},l}) \notag \\
& \,= I(X^{l-1}; Y_{\overline{i}}^{l-1} | C_i^{l-1}, Y_{\overline{i}}^{l-1}) 
  + I(X_l; Y_{i,l} | Y_{\overline{i},l}),
\end{align}
where $(a)$ follows from \eqref{(24)} and \eqref{(25)}, and the trivial inequalities\linebreak $H(Y_i^{l-1} | C_i^{l-1}, Y_{\overline{i}}^{l}) 
\leq H(Y_i^{l-1} | C_i^{l-1}, Y_{\overline{i}}^{l-1})$, and\linebreak $H(Y_{i,l} | C_i^{l-1}, Y_i^{l-1}, Y_{\overline{i}}^{l}) 
\leq H(Y_{i,l} | Y_{\overline{i},l})$. Then applying \eqref{main} to the above term for $l = n, n-1, \cdots, 1$, we have:
\begin{align*}
    I(X_l; Y_{i,l} | Y_{\overline{i},l})&\leq n.\max_{P_X}I(X; Y_i | Y_{\overline{i}})-(n-1).\max_{P_X} I(X; Y_i | Y_{\overline{i}})\\
    & = \max_{P_X} I(X; Y_i | Y_{\overline{i}}).
\end{align*}
If we take $Y_{\overline{i}}$ to be constant, then:\linebreak $I(X_l;Y_{i,l})\leq\max_{P_X} I(X;Y_i)$.

\parbreak The whole of the above calculation can be repeated for the sum rate, then we have: 
\[ R_1 + R_2 \leq \lim_{n\to\infty} \frac{k_1+k_2}{n}\leq\max_{P_X} I(X;Y_1,Y_2).\]

\parbreak To derive another upper bound on the OT capacity, we have the following chain of equalities and inequalities:
\begin{lemma}\label{lemma: Rudolf-app}
For any $z_i\in \{0, 1\}$ and $i\in\{1,2\}$, we have: 
\[
I(M_{iz_i}; Y_i^n, R_{B_i} | X^n, \mathbf{C}_i,V_{\overline{i}}, Z_i = z_i) = 0,
\]
\end{lemma}
\begin{proof}
    Define for $i \in \{1,2\}$, $l \in [1, n]$, $t \in [1, r_l]$, $C_{1, l, 1:j} \triangleq (C_{0, i,l}(j), C_{i, l}(j))_{j \in [1, t]}$ as the messages exchanged between Alice and Bob-$i$ between the first and the $j$-th communication rounds occurring after the $t$-th channel usage. Let $C_i^l \triangleq (C_{i, l, 1:r})_{j \in [1, l]}$ represent all messages exchanged between Alice and Bob-$i$ before the $l+1$-th channel use. Let $l \in [1, n]$ and $j \in [1, r_l]$. Then, we have
\begin{align}\label{eq 1}
    I(M_{i0}&, M_{i1},  R_A; Y_i^l, R_{B_i} | X^l, C_i^{l-1},C_{i,l,1:j}, Z_i, V_{\overline{i}})\notag\\
    & \overset{(a)}{\leq} I(M_{i0}, M_{i1}, R_A; Y_i^l, R_{B_i}, C_{0,i,l}(j) | X^l, C_i^{l-1},C_{i,l,1:(j-1)}, C_{i,l}(j), Z_i, V_{\overline{i}})\\ 
    & \overset{(b)}{=} I(M_{i0}, M_{i1}, R_A; Y_i^l, R_{B_i} | X^l, C_i^{l-1},C_{i,l,1:(j-1)}, C_{i,l}(j), Z_i, V_{\overline{i}})\notag\\
    & \leq I(M_{i0}, M_{i1}, R_A, C_{i,l}(j); Y_i^l, R_{B_i} | X^l, C_i^{l-1},C_{i,l,1:(j-1)},  Z_i ,V_{\overline{i}})\notag\\
    &\overset{(c)}{=} I(M_{i0}, M_{i1}, R_A; Y_i^l, R_{B_i} | X^l, C_i^{l-1},C_{i,l,1:(j-1)}, Z_i, V_{\overline{i}})\label{eq 2} \\
    &\overset{(d)}{\leq} I(M_{i0}, M_{i1}, R_A; Y_i^{l}, R_{B_i} | X^l, C_i^{l-1}, Z_i, V_{\overline{i}}) \notag\\
    &\overset{(e)}{=} I(M_{i0}, M_{i1}, R_A ; Y_i^{l-1}, R_{B_i} | X^l, C_i^{l-1}, Z_i, V_{\overline{i}}) \notag\\
    &\leq I(M_{i0}, M_{i1}, R_A, X_{l}; Y_i^{l-1}, R_{B_i} | X^{l-1}, C_i^{l-1}, Z_i, V_{\overline{i}}) \notag\\
    &\overset{(f)}{=} I(M_{i0}, M_{i1}, R_A; Y_i^{l-1}, R_{B_i} | X^{l-1}, C_i^{l-1}, Z_i, V_{\overline{i}})\label{eq 3} = 0.    
\end{align}

The steps are justified as follows:
\begin{itemize}
    \item $(a)$ follows by the definition of $C_{i, l, 1:(j)} = (C_{i, l, 1:(j-1)}, C_{0, i,l}(j), C_{1, l}(j))$ and by the chain rule.
    \item $(b)$ follows by the chain rule and because $C_{0, i,l}(j)$ depends on\linebreak $(C_{i, l, 1:(j-1)}, Z_i, V_{\overline{i}}, R_{B_i}, Y_i^l, C_i^{l-1}, C_{i,l}(j))$.
    \item $(c)$ follows by the chain rule and because $C_{i,l}(j)$ depends on\linebreak $(M_{i0}, M_{i1}, R_A, C_{i, l, 1:(j-1)}, C_i^{l-1})$.
    \item $(d)$ follows by previous iterations $(j-1)$ of \eqref{eq 1} to \eqref{eq 2}.
    \item $(e)$ follows from the Markov chain:\linebreak $(M_{i0},M_{i1},R_A)-(Y_i^{l-1}, R_{B_i}, X^l, C_i^{l-1}, Z_i, X^l, V_{\overline{i}})-Y_{i,l}$.
    \item $(f)$ follows from the chain rule and the fact that $X_{l}$ is a function of\linebreak $(M_{i0}, M_{i1}, R_A, C_i^{l-1})$.
\end{itemize}

A straightforward calculation gives, for any $l \in [1, n]$, we have

\[
I(M_{i0}, M_{i1}, R_A; Y_i^l, R_{B_i} | X^l, C_i^{l}, Z_i, V_{\overline{i}}) = 0.
\]

This completes the proof.
\end{proof}
\begin{lemma}\label{lemma: smooth}
For any $z_i\in \{0, 1\}$, we have:
    \begin{equation*}
        H(M_{i\overline{Z}_i}|X^n, \mathbf{C}_i, V_{\overline{i}},Z_i = z_i) = o(n).
    \end{equation*}
\end{lemma}
\begin{proof}
    \begin{align*}
        H(M_{i\overline{z}_i}|X^n, \mathbf{C}_i, V_{\overline{i}},Z_i = z_i) & \stackrel{(a)}{\leq} H(M_{i\overline{z}_i}|X^n, \mathbf{C}_i, V_{\overline{i}},Z_i = \overline{z}_i) + o(n)\\
        & \stackrel{(b)}{=} H(M_{i\overline{z}_i}|Y_i^n, R_{B_i} , X^n, \mathbf{C}_i, V_{\overline{i}},Z_i = \overline{z}_i) + o(n)\\
        & \stackrel{(c)}{\leq} H(M_{i\overline{z}_i}|Y_i^n, R_{B_i} , \hat{M}_{i\overline{z}_i}, \mathbf{C}_i, V_{\overline{i}},Z_i = \overline{z}_i) + o(n)\\
        & \stackrel{(d)}{\leq} o(n).
    \end{align*}
    Finally, 
    \begin{align*}
        H(M_{i\overline{Z}_i}|X^n, \mathbf{C}_i, V_{\overline{i}},Z_i = z_i) & = \sum_{z_i} \mathbb{P} (Z_i = z_i) \times H(M_{i\overline{z}_i}|X^n, \mathbf{C}_i, V_{\overline{i}},Z_i = z_i)\\
        & = o(n),
    \end{align*}
    where $(a)$ follows from Eq. \eqref{smoothed: applied lemma}, $(b)$ follows from Lemma \ref{lemma: Rudolf-app}, $(c)$ holds because $\hat{M}_{i\overline{z}_i}$ is a function of $(Y_i^n, R_{B_i}, \mathbf{C}_i)$, and $(d)$ holds by Fano's inequality and Eq. \eqref{goals: BC-Coll-1}.
\end{proof}
\parbreak Using Eq. \eqref{smoothed1} and Lemma \ref{lemma: smooth}, we have:
\begin{align}
           k_i = H(M_{i\overline{z}_i}|Z_i=z_i) & \,= H(M_{i\overline{z}_i}| R_{B_i}, Y_i^n, \mathbf{C}_i, V_{\overline{i}},Z_i = z_i) + o(n)\notag\\
           & \,= H(M_{i\overline{z}_i}| R_{B_i}, Y_i^n, \mathbf{C}_i, V_{\overline{i}},Z_i = z_i)\notag\\
           & \quad+ H(X^n|M_{i\overline{z}_i},R_{B_i}, Y_i^n, \mathbf{C}_i, V_{\overline{i}},Z_i = z_i)\notag\\
           & \quad - H(X^n|M_{i\overline{z}_i},R_{B_i}, Y_i^n, \mathbf{C}_i, V_{\overline{i}},Z_i = z_i) + o(n)\notag\\
           & \,= H(M_{i\overline{z}_i},X^n|R_{B_i}, Y_i^n, \mathbf{C}_i, V_{\overline{i}},Z_i = z_i)\notag\\
           & \quad- H(X^n|M_{i\overline{z}_i}R_{B_i}, Y_i^n, \mathbf{C}_i, V_{\overline{i}},Z_i = z_i) + o(n)\notag\\
           & \,\leq H(M_{i\overline{z}_i},X^n|R_{B_i}, Y_i^n, \mathbf{C}_i, V_{\overline{i}},Z_i = z_i) + o(n)\notag\\
           & \,= H(X^n|R_{B_i}, Y_i^n, \mathbf{C}_i, V_{\overline{i}},Z_i = z_i)\notag\\
           & \quad + H(M_{i\overline{z}_i}|X^n, R_{B_i}, Y_i^n, \mathbf{C}_i, V_{\overline{i}},Z_i = z_i) + o(n)\notag\\
           & \,\leq H(X^n|R_{B_i}, Y_i^n, \mathbf{C}_i, V_{\overline{i}},Z_i = z_i)\notag\\
           & \quad + H(M_{i\overline{z}_i}|X^n, \mathbf{C}_i, V_{\overline{i}},Z_i = z_i) + o(n)\notag\\ 
           & \stackrel{(a)}{=}  H(X^n|R_{B_i}, Y_i^n, \mathbf{C}_i, V_{\overline{i}},Z_i = z_i) + o(n)\notag\\
           & \,= H(X^n|R_{B_i}, Y_i^n, \mathbf{C}_i, R_{B_{\overline{i}}}, Y_{\overline{i}}^n, \mathbf{C}_{\overline{i}},Z_i = z_i, Z_{\overline{i}} = z_{\overline{i}}) + o(n)\notag\\
           & \stackrel{(b)}{\leq} H(X^n|Y_i^n,Y_{\overline{i}}^n,Z_i=\overline{z}_i, Z_{\overline{i}} = z_{\overline{i}}) + o(n)\notag\\
           & \,\leq \sum_{l=[1:n]} H(X_{l}|Y_{i,l}, Y_{\overline{i},l}, Z_i=\overline{z}_i, Z_{\overline{i}} = z_{\overline{i}}) + o(n)\notag\\
           & \stackrel{(c)}{\leq} \sum_{l=1}^{n} H(X_{l}|Y_{i,l}, Y_{\overline{i},l}) + o(n), 
       \end{align}
       where $(a)$ follows from Lemma \ref{lemma: smooth}, $(b)$ is because conditioning does not increase the entropy, and $(c)$ follows from the arguments presented in \cite{Rudolf1}.

\section{Proof of Lemma \ref{lem: cri. nColl-BC}}\label{app: proof of lem: cri. nColl-BC}
Conditions \eqref{goals: BC-nColl-1}, \eqref{goals: BC-nColl-2}, and \eqref{goals: BC-nColl-3} can be proved similar to those of conditions \eqref{goals: BC-Coll-1}, \eqref{goals: BC-Coll-2}, and \eqref{goals: BC-Coll-4}. See \ref{app: proof of lem: cri. Coll-BC}.

\parbreak As previously noted, the key distinction between Protocol~1 and Protocol~2 lies in the fact that the latter is not a pairwise protocol. Specifically, an OT protocol is first executed with one of the receivers; if $p_i > \tfrac{1}{2}$, a subsequent OT protocol is then executed with the other receiver. In this setting, $\mathbf{C}_1$ and $\mathbf{C}_2$ carry independent information. This corresponds to the following Markov chain:
\[
Z_i \;-\; M_{i0}, M_{i1}, X^n, Y_{\overline{i}}, \mathbf{C}_i \;-\; Z_{\overline{i}}, M_{\overline{i}0}, M_{\overline{i}1}, \mathbf{C}_{\overline{i}} \, .
\]

Now, let us examine the security of Protocol~1 in the presence of colluding parties. Clearly, conditions~\eqref{goals: BC-Coll-3} and~\eqref{goals: BC-Coll-5} must also be satisfied. If collusion is permitted, then Bob-$i$ can disclose his sets of erased and non-erased indices to Bob-$\overline{i}$. This exchange constitutes the primary source of insecurity in Protocol~1 under collusion. Formally, we obtain:
\begin{align}\label{eq: dep}
    I (Z_i; U, V_{\overline{i}}) & = I(Z_i;M_{i0},M_{i1},M_{\overline{i}0},M_{\overline{i}1}, Z_{\overline{i}}, X^n, Y_{\overline{i}}^n, \mathbf{C})\notag\\
    & = I(Z_i;M_{i0},M_{i1},M_{\overline{i}0},M_{\overline{i}1}, Z_{\overline{i}}, X^n, Y_{\overline{i}}^n, \mathbf{C}_i,\mathbf{C}_{\overline{i}})\notag\\
    & = I(Z_i;M_{i0},M_{i1},X^n, Y_{\overline{i}}^n, \mathbf{C}_i)\notag\\
           & \quad + I(Z_i;Z_{\overline{i}},M_{\overline{i}0},M_{\overline{i}1},\mathbf{C}_{\overline{i}}|M_{i0},M_{i1},X^n,Y_{\overline{i}},\mathbf{C}_i).
    \end{align}
The first term equal to zero due to \eqref{eq: perfect}. But the second term is not equal to zero because of the structure of Protocol 1 since $\{S_{i0},S_{i1}\}\subset\mathbf{C}_i$ and both of the receivers have access to $S_{i0},S_{i1}, i\in\{1,2\}$. The same situation holds for \eqref{goals: BC-Coll-5}:
\begin{align}
        I (M_{i0},M_{i1},M_{\overline{i}\overline{Z}_{\overline{i}}}, Z_i; V_{\overline{i}})_{i\in\{1,2\}}&\,= I (M_{i0},M_{i1}, Z_i; V_{\overline{i}})\notag\\
        & \quad + I (M_{\overline{i}\overline{Z}_{\overline{i}}};V_{\overline{i}}|M_{i0},M_{i1}, Z_i)\notag\\
        &\,= I (M_{i0},M_{i1}, Z_i; Z_{\overline{i}},Y^n_{\overline{i}},\mathbf{C})\notag\\
        & \quad + I (M_{\overline{i}\overline{Z}_{\overline{i}}};V_{\overline{i}},M_{i1},M_{i2}, Z_i)\notag\\
        &\,= I (M_{i0},M_{i1}, Z_i; Z_{\overline{i}},Y^n_{\overline{i}},\mathbf{C}_i)\notag\\
        & \quad + I (M_{\overline{i}\overline{Z}_{\overline{i}}};V_{\overline{i}},M_{i0},M_{i1}, Z_i )\notag\\
        & \quad+ I (M_{i0},M_{i1}, Z_i;\mathbf{C}_{\overline{i}}| Z_{\overline{i}},Y^n_{\overline{i}},\mathbf{C}_i).
\end{align}
The first term vanishes as $n \to \infty$, while the second term is equal to zero, as shown in~\eqref{eq: dep2}. However, the last term does not vanish, for the same reason discussed above.
  
\section{Proof of Theorem \ref{thm: OT capacity-non-col}}\label{app: thm: OT capacity-non-col}
  We must bound three conditional min-entropies to get three bounds on the lower bound:
        \begin{equation}\label{eq: first}
            H^\varepsilon_{\infty}(\mathbf{X}|_{S_{1\overline{Z}_1}}|h_{10}(R_{10}, \mathbf{X}|_{S_{10}}), h_{11}(R_{11}, \mathbf{X}|_{S_{11}}), Y_1^n, R^{(1)}, T^{(1)}),
        \end{equation}
        \begin{equation}\label{eq: second}
            H^\varepsilon_{\infty}(\mathbf{X}|_{S_{2\overline{Z}_2}}|h_{20}(R_{20}, \mathbf{X}|_{S_{20}}), h_{21}(R_{21}, \mathbf{X}|_{S_{21}}), Y_2^n, R^{(2)}, T^{(2)}),
        \end{equation}
        \begin{align}\label{eq: third}
            H^\varepsilon_{\infty}(\mathbf{X}|_{S_{1\overline{Z}_1},S_{2\overline{Z}_2}}|&h_{10}(R_{10}, \mathbf{X}|_{S_{10}}), h_{11}(R_{11}, \mathbf{X}|_{S_{11}}),\notag\\ & h_{20}(R_{20}, \mathbf{X}|_{S_{20}}), h_{21}(R_{21}, \mathbf{X}|_{S_{21}}), Y_1^n,Y_2^n, R).
        \end{align}
From \cite{Winter1} we know that the first two expressions above smooth min-entropies can be bounded from below as:
\begin{equation}
    H^\varepsilon_{\infty}(\mathbf{X}|_{S_{i\overline{Z}_i}}|h_{i0}(R_{i0}, \mathbf{X}|_{S_{i0}}), h_{i1}(R_{i1}, \mathbf{X}|_{S_{i1}}), Y_i^n, R^{(i)}, T^{(i)})\geq pn H(X) - s_i n -\delta n,
\end{equation}
for $i\in\{1,2\}$, where $\varepsilon = 2^{-\alpha n}$ and  $\varepsilon' = 2^{-\alpha' n}$ ($\varepsilon$ and $\varepsilon'$ are negligible in $n$) and $\delta \geq (\alpha+\alpha' + 2\eta +4\sqrt{\alpha})> 0$.

\parbreak 
Now consider \eqref{eq: third}. From the i.i.d. property of the channel, we have:
  \begin{align*}
            H^\varepsilon_{\infty}&(\mathbf{X}|_{S_{1\overline{Z}_1},S_{2\overline{Z}_2}}|h_{10}(R_{10}, \mathbf{X}|_{S_{10}}), h_{11}(R_{11}, \mathbf{X}|_{S_{11}}),\\
            & \qquad\qquad\qquad\qquad\qquad\qquad\qquad h_{20}(R_{20}, \mathbf{X}|_{S_{20}}), h_{21}(R_{21}, \mathbf{X}|_{S_{21}}), Y_1^n,Y_2^n, R)\\
            & = H^\varepsilon_{\infty}(\mathbf{X}|_{S_{1\overline{Z}_1},S_{2\overline{Z}_2}}| h_{1j}(R_{1j},\mathbf{X}|_{S_{1\overline{Z}_1}}),   h_{2j}(R_{2j},\mathbf{X}|_{S_{2\overline{Z}_2}}),\mathbf{Y}_1|_{S_{1\overline{Z}_1}},\mathbf{Y}_2|_{S_{2\overline{Z}_2}}, R).
        \end{align*}
        
        Applying \eqref{eq: holenstein-1} and \eqref{eq: holenstein-1 + holenstein-1} for $\varepsilon, \varepsilon'>0$, we have:
        \begin{align}\label{eq: after applying holenstein-1}
        &H^{\varepsilon+\varepsilon'}_{\infty}(\mathbf{X}|_{S_{1\overline{Z}_1},S_{2\overline{Z}_2}}| h_{1j}(R_{1j},\mathbf{X}|_{S_{1\overline{Z}_1}}),   h_{2j}(R_{2j},\mathbf{X}|_{S_{2\overline{Z}_2}}),\mathbf{Y}_1|_{S_{1\overline{Z}_1}},\mathbf{Y}_2|_{S_{2\overline{Z}_2}}, R)\\
        & \geq H_\infty^\varepsilon (\mathbf{X}|_{S_{1\overline{Z}_1},S_{2\overline{Z}_2}},h_{1j}(R_{1j}, \mathbf{X}|_{S_{1\overline{Z}_1}}),h_{2j}(R_{2j}, \mathbf{X}|_{S_{2\overline{Z}_2}})|\mathbf{Y}_1|_{S_{1\overline{Z}_1}},\mathbf{Y}_2|_{S_{2\overline{Z}_2}}, R)\notag\\
        & \quad- H_0(h_{1j}(R_{1j}, \mathbf{X}|_{S_{1\overline{Z}_1}}),h_{2j}(R_{2j}, \mathbf{X}|_{S_{2\overline{Z}_2}})|\mathbf{Y}_1|_{S_{1\overline{Z}_1}},\mathbf{Y}_2|_{S_{2\overline{Z}_2}}, R) - \log(\frac{1}{\varepsilon'}).
        \end{align}
Note that $H_0(h_{1j}(R_{1j}, \mathbf{X}|_{S_{1\overline{Z}_1}}),h_{2j}(R_{2j}, \mathbf{X}|_{S_{2\overline{Z}_2}})|\mathbf{Y}_1|_{S_{1\overline{Z}_1}},\mathbf{Y}_2|_{S_{2\overline{Z}_2}}, R)$ limits the number of distinct possible outputs, restricting the amount of information Alice can gain about receivers' choices: 
        \[
        H_0(h_{1j}(R_{1j}, \mathbf{X}|_{S_{1\overline{Z}_1}}),h_{2j}(R_{2j}, \mathbf{X}|_{S_{2\overline{Z}_2}})|\mathbf{Y}_1|_{S_{1\overline{Z}_1}},\mathbf{Y}_2|_{S_{2\overline{Z}_2}}, R)\leq s_1 n + s_2 n.
        \]
        
Equation \eqref{eq: after applying holenstein-1} simplifies to:
\begin{align}\label{eq: middle}
    H&^{\varepsilon+\varepsilon'}_{\infty}(\mathbf{X}|_{S_{1\overline{Z}_1},S_{2\overline{Z}_2}}| h_{1j}(R_{1j},\mathbf{X}|_{S_{1\overline{Z}_1}}),   h_{2j}(R_{2j},\mathbf{X}|_{S_{2\overline{Z}_2}}),\mathbf{Y}_1|_{S_{1\overline{Z}_1}},\mathbf{Y}_2|_{S_{2\overline{Z}_2}}, R)\notag\\
    & \geq H_\infty^\varepsilon (\mathbf{X}|_{S_{1\overline{Z}_1},S_{2\overline{Z}_2}},h_{1j}(R_{1j}, \mathbf{X}|_{S_{1\overline{Z}_1}}),h_{2j}(R_{2j}, \mathbf{X}|_{S_{2\overline{Z}_2}})|\mathbf{Y}_1|_{S_{1\overline{Z}_1}},\mathbf{Y}_2|_{S_{2\overline{Z}_2}}, R)\notag\\
    & \qquad\qquad\qquad\qquad\qquad\qquad\qquad\qquad\qquad\qquad\qquad\qquad\quad- \log(\frac{1}{\varepsilon'}),
\end{align}
Consider the following quantity. Putting $U\triangleq (\mathbf{X}|_{S_{1\overline{Z}_1},S_{2\overline{Z}_2}})$,\linebreak $ V \triangleq (h_{1j}(R_{1j}, \mathbf{X}|_{S_{1\overline{Z}_1}}), h_{2j}(R_{2j}, \mathbf{X}|_{S_{2\overline{Z}_2}}))$, and $W \triangleq (\mathbf{Y}_1|_{S_{1\overline{Z}_1}},\mathbf{Y}_2|_{S_{2\overline{Z}_2}}, R)$, from \eqref{eqeq}, we have ($\varepsilon, \varepsilon'>0, \varepsilon = \varepsilon + \varepsilon' + 2\varepsilon''$):
\begin{align}\label{epsilon'}
     H&^{\varepsilon+\varepsilon'+2\varepsilon''}_{\infty}(\mathbf{X}|_{S_{1\overline{Z}_1},S_{2\overline{Z}_2}}, h_{1j}(R_{1j}, \mathbf{X}|_{S_{1\overline{Z}_1}}), h_{2j}(R_{2j}, \mathbf{X}|_{S_{2\overline{Z}_2}})| \mathbf{Y}_1|_{S_{1\overline{Z}_1}},\mathbf{Y}_2|_{S_{2\overline{Z}_2}}, R)\notag\\
        & \geq H^{\varepsilon'}_{\infty}(\mathbf{X}|_{S_{1\overline{Z}_1},S_{2\overline{Z}_2}}| \mathbf{Y}_1|_{S_{1\overline{Z}_1}},\mathbf{Y}_2|_{S_{2\overline{Z}_2}}, R)\notag\\
            & \quad + H^{\varepsilon''}_{\infty}(h_{1j}(R_{1j}, \mathbf{X}|_{S_{1\overline{Z}_1}}), h_{2j}(R_{2j}, \mathbf{X}|_{S_{2\overline{Z}_2}})|\mathbf{X}|_{S_{1\overline{Z}_1},S_{2\overline{Z}_2}}, \mathbf{Y}_1|_{S_{1\overline{Z}_1}},\mathbf{Y}_2|_{S_{2\overline{Z}_2}}, R)\notag\\
            & \quad-\log(\frac{1}{\gamma})\notag\\
            & = H^{\varepsilon'}_{\infty}(\mathbf{X}|_{S_{1\overline{Z}_1},S_{2\overline{Z}_2}}| \mathbf{Y}_1|_{S_{1\overline{Z}_1}},\mathbf{Y}_2|_{S_{2\overline{Z}_2}}, R) - \log(\frac{1}{\gamma}),
        \end{align}
where $\gamma = \varepsilon - \varepsilon' - 2\varepsilon''$. Recall \eqref{eq: middle}. Putting $\varepsilon = \varepsilon + \varepsilon' + 2\varepsilon''$ and Substituting \eqref{epsilon'} to it yields:
\begin{align}\label{epsilon''}
    H&^{\varepsilon+2\varepsilon'+2\varepsilon''}_{\infty}(\mathbf{X}|_{S_{1\overline{Z}_1},S_{2\overline{Z}_2}}|h_{1j}(R_{1j}, \mathbf{X}|_{S_{1\overline{Z}_1}}), h_{2j}(R_{2j}, \mathbf{X}|_{S_{2\overline{Z}_2}}), \mathbf{Y}_1|_{S_{1\overline{Z}_1}},\mathbf{Y}_2|_{S_{2\overline{Z}_2}}, R)\notag\\
    & \geq H^{\varepsilon'}_{\infty}(\mathbf{X}|_{S_{1\overline{Z}_1},S_{2\overline{Z}_2}}| \mathbf{Y}_1|_{S_{1\overline{Z}_1}},\mathbf{Y}_2|_{S_{2\overline{Z}_2}}, R) -(s_1+s_2)n - \log(\frac{1}{\lambda}) - \log(\frac{1}{\gamma}),
\end{align}

where $\lambda\triangleq \varepsilon + \varepsilon' + 2\varepsilon'', \varepsilon' = \varepsilon + \varepsilon' + 2\varepsilon'', \gamma = \varepsilon - \varepsilon' - 2\varepsilon''$. Let $V_i, i\in\{1,2\}$ be an i.i.d. random variable so that $V_i = e$ (erasure) with probability $\frac12-\eta$ and $V_i = Q_i$ (the output of channel $W$ on input $X$). With negligible error probability and $W$ being i.i.d., for $S_{1\overline{Z}_1}$ and $S_{2\overline{Z}_2}$, we have:
\begin{align}\label{eq: middle+1}
    H_\infty^{\varepsilon+\varepsilon'}&(\mathbf{X}|_{ S_{1\overline{Z}_1},S_{2\overline{Z}_2}}|\mathbf{Y}_1|_{S_{1\overline{Z}_1}},\mathbf{Y}_2|_{S_{2\overline{Z}_2}})\notag\\
    &\,\geq H_\infty^{\varepsilon}(\mathbf{X}|_{\lvert S_{1\overline{Z}_1},S_{2\overline{Z}_2}\lvert}|\mathbf{V}_1|_{\lvert S_{1\overline{Z}_1}\lvert},\mathbf{V}_2|_{\lvert S_{2\overline{Z}_2}\lvert})\notag\\
    & \overset{(a)}{\geq} \lvert S_{1\overline{Z}_1},S_{2\overline{Z}_2}\rvert H(X|V_1,V_2)-4\sqrt{\lvert S_{1\overline{Z}_1},S_{2\overline{Z}_2}\rvert\log(1/\varepsilon)}\log \lvert \mathcal{X} \rvert\notag\\
    & \,\geq (p-\eta)n H(X|V_1,V_2)-4\sqrt{((1-p_1)(1-p_2)-\eta)n\log(1/\varepsilon)}\notag\notag\\
    & \,\geq pn H(X|V_1,V_2)-\eta n H(X|V_1,V_2)-4\sqrt{((1-p_1)(1-p_2)-\eta)n\log(1/\varepsilon)}\notag\notag\\
    & \,\geq pn H(X|V_1,V_2)-\eta n -4\sqrt{(p-\eta)n\log(1/\varepsilon)}\notag\notag\\
    & \overset{(b)}{\geq} (1-p_1)(1-p_2)n (1-2\eta)H(X)+2\eta n H(X|V_1,V_2)-\eta n\notag\\
    & \qquad\qquad\qquad\qquad\qquad\qquad\qquad\quad\,-4\sqrt{((1-p_1)(1-p_2)-\eta)n\log(1/\varepsilon)}\notag\notag\\
    & \,\geq (1-p_1)(1-p_2)n H(X)-2\eta n-4\sqrt{n\log(1/\varepsilon)},
\end{align}
where $(a)$ follows from Lemma \ref{lemm: Holenstein}, $\lvert\mathcal{X}\rvert = 2$, and $(b)$ follows from the fact that honest Bob-$i$ doesn't split the erasures received from Alice between $S_{i0}$ and $S_{i1}$, with probability exponentially close to one, the total number of non-erased symbols Bob-$i$ receives from the sender will not exceed $(1-p_i + \eta)n$, so the number of non-erasures in $\mathbf{Y}_i|_{S_{i\overline{Z}_i}}$ is at most $\lvert(1-p_i-\eta)n - (1-p_i + \eta)n\rvert = 2n\eta$.

\parbreak Putting $\varepsilon = \varepsilon + \varepsilon'+2\varepsilon''$ in \eqref{eq: middle}, then putting \eqref{eq: middle+1} to \eqref{eq: middle}, we have: 
\begin{align}
    H&^{\varepsilon+2\varepsilon'+2\varepsilon''}_{\infty}(\mathbf{X}|_{S_{1\overline{Z}_1},S_{2\overline{Z}_2}}| h_{1j}(R_{1j},\mathbf{X}|_{S_{1\overline{Z}_1}}),   h_{2j}(R_{2j},\mathbf{X}|_{S_{2\overline{Z}_2}}),\mathbf{Y}_1|_{S_{1\overline{Z}_1}},\mathbf{Y}_2|_{S_{2\overline{Z}_2}}, R)\notag\\
    &\,\geq (1-p_1)(1-p_2)n H(X)-2\eta n-4\sqrt{n\log(1/\varepsilon)}-s_1 n - s_2 n - \log(\frac{1}{\lambda})\notag\\
    & \qquad\qquad\qquad\qquad\qquad\qquad\qquad\qquad\qquad\qquad\qquad\qquad\qquad\,\,\,\,\,- \log(\frac{1}{\gamma})\notag\\
    & \overset{(a)}{=} (1-p_1)(1-p_2)n H(X) - s_1 n -s_2 n -2\eta n-4\sqrt{n\alpha''}-n (\alpha+\alpha'),
\end{align}
for any $\delta \geq (2\alpha + \alpha'+ 2\eta +4\sqrt{\alpha''})> 0$,  $(a)$ follows from setting $\gamma = 2^{-\alpha n}, \varepsilon'' = 2^{-n\alpha''}$ and  $\lambda = 2^{-\alpha' n}$.

\parbreak Due to the Chernoff bound, we know that the probability that Bob aborts the protocol in step $(2)$ tends to zero as $n\rightarrow \infty$. Protocol 1 fails in step $(5)$, if there is more than one, such as $(\hat{\mathbf{x}}|_{S_{iZ_i}})$ where $h_i(\mathbf{X}|_{S_{iZ_i}}) = h_i(\hat{\mathbf{X}}|_{S_{iZ_i}})$. We know that if all players are honest, then $(\mathbf{Q}_1|_{S_{1Z_1}},\mathbf{Q}_2|_{S_{2Z_2}}) =( \mathbf{Y}_1|_{S_{1Z_1}},\mathbf{Y}_2|_{S_{2Z_2}})$ with probability exponentially close to one and the number of sequences $\hat{\mathbf{x}}|_{S_{1Z_1},S_{2Z_2}}$ jointly typical with $(\mathbf{q}_1|_{S_{1Z_1}},\mathbf{q}_2|_{S_{2Z_2}})$ can be bounded from above so that if the sequence $\hat{\mathbf{x}}|_{S_{1Z_1},S_{2Z_2}}$ is not typical with$(\mathbf{q}_1|_{S_{1Z_1}},\mathbf{q}_2|_{S_{2Z_2}})$:\linebreak  $2^{\lvert S_{1Z_1} \rvert\lvert S_{2Z_2} \rvert(H(X,Q_1,Q_2)- H(Q_1,Q_2)+ \delta'}) =  2^{(1-p_1)(1-p_2)n(H(X,Q_1,Q_2)- H(Q_1,Q_2)+ \delta')}$,\linebreak $\delta'>0$. From \eqref{eq: hashs-double}, we know that\linebreak $p\leq 2^{-(s_1+s_2)n}2^{(1-p_1)(1-p_2)n(H(X,Q_1,Q_2)- H(Q_1,Q_2)+ \delta')}$, then $s_1+s_2 > (1-p_1)(1-p_2) (H(X,Q_1,Q_2)- H(Q_1,Q_2)+ \delta')$.

\parbreak The final lower bound is as follows:
\begin{align*}
    r_1 &<  \max_{P_{X}} I(X;Y_1),\\
    r_2 &<  \max_{P_{X}} I(X;Y_2),\\
    r_1+r_2&<\max_{P_{X}} I(X;Y_1,Y_2) .
\end{align*}
The lower bound coincides with the upper bound presented in Theorem \ref{thm:BC-ncoll}. This completes the proof.

\section{Proof of Lemma \ref{lem: cri. Coll-BC}}\label{app: proof of lem: cri. Coll-BC}

When $p_i\leq \frac12$, then the proof is the same as \cite{Rudolf1}, because if $p_i\leq\frac12$, then there are not enough erased bits to build up secret keys for the second OT link from the same source $X$. 
    Now consider the case $p_i>\frac12$. Let $(\mathcal{P}_n)_{n\in\mathbb{N}}$ be a sequence of Protocol 1 instances: 
\begin{itemize}
    \item Due to the Chernoff bound, we know that the probability that Bob aborts the protocol-$i$ in step 2 tends to zero as $n\rightarrow
    \infty$. When $|\overline{E}|<r_in$ and $|E|<r_in$, then Bob-$i$ knows $\mathbf{X}|_{S_{Z_i}}$. Since Bob-$i$ also knows $\kappa_{Z_i}$, Bob-$i$ can compute the key $\kappa_{Z_i}(\mathbf{X}|_{S_{Z_i}})$. Then Bob-$i$ can recover $M_{iZ_i}$ from $M_{iZ_i}\oplus\kappa_{Z_i}(\mathbf{X}|_{S_{Z_i}})$ sent by Alice. Then, \[
    \lim_{n\rightarrow\infty}\text{Pr}\,\left[(\hat{M}_{iZ_i})\neq (M_{iZ_i})\right]_{i\in\{1,2\}}= 0\\
    \]

    \item Since the probability that Bob aborts the protocol in step 2 tends to zero as $n \rightarrow \infty$, then we have:
\begin{align}\label{eq: mid}
      I (  M_{1\overline{Z}_1}&,M_{2\overline{Z}_2};V_1,V_2) \notag\\
      & = I (  M_{1\overline{Z}_1},M_{2\overline{Z}_2};Z_1, Z_2, Y_1^n, Y_2^n, \mathbf{C})\notag\\
      & = I (  M_{1\overline{Z}_1},M_{2\overline{Z}_2};Z_1, Z_2, Y_1^n, Y_2^n, \mathbf{C}_1, \mathbf{C}_2)\notag\\
      & = I (  M_{1\overline{Z}_1};Z_1, Y_1^n, Y_2^n,  \mathbf{C}_1) + I (  M_{2\overline{Z}_2};Z_1, Z_2, Y_1^n, Y_2^n,  \mathbf{C}_1, \mathbf{C}_2|M_{1\overline{Z}_1})\notag\\
      & = I (  M_{1\overline{Z}_1};Z_1, Y_1^n, Y_2^n,  \mathbf{C}_1) + I ( M_{2\overline{Z}_2};M_{1\overline{Z}_1}, Z_1, Z_2, Y_1^n, Y_2^n, \mathbf{C}_1,\mathbf{C}_2)\notag\\
      & \overset{(a)}{=} I (  M_{1\overline{Z}_1};Z_1, Y_1^n, Y_2^n,  \mathbf{C}_1) + I ( M_{2\overline{Z}_2}; Z_2, Y_1^n, Y_2^n, \mathbf{C}_2)\notag\\
      & \overset{(b)}{=} I (  M_{1\overline{Z}_1};Z_1, Y_1^n, Y_2^n,  \mathbf{C}_1) + I (  M_{2\overline{Z}_2};Z_2, Y_2^n|_{S'}, S', \mathbf{C}_2),
    \end{align}
    where $(a)$ follows from the independence of $(M_{1\overline{Z}_1},Z_1, \mathbf{C}_1)$ from $M_{2\overline{Z}_2}$ given $(Z_2, Y_1^n,Y_2^n, S', \mathbf{C}_2)$, and $(b)$ follows immediately from the second phase of Protocol \ref{protocol: BC-Coll}. 
    
    \parbreak Now, it suffices to prove that both terms tend to zero as $n\to\infty$. Consider the first term of \eqref{eq: mid}:
\begin{align}\label{final2: BC-nColl}
        I &(  M_{1\overline{Z}_1};Z_1, Y_1^n, Y_2^n,  \mathbf{C}_1)\notag\\
        & \overset{(a)}{=} I (  M_{1\overline{Z}_1};Z_1, Y_1^n, Y_2^n,  S_0,S_1, \kappa_{10}, \kappa_{11}, M_{10}\oplus\kappa_{10} (\mathbf{X}|_{S_0}), M_{11}\oplus\kappa_{11}(\mathbf{X}|_{S_1})\notag\\
        & \, = I (  M_{1\overline{Z}_1};Z_1, Y_1^n, Y_2^n,  S_{Z_1},S_{\overline{Z}_1}, \kappa_{1Z_1}, \kappa_{1\overline{Z}_1}, M_{1Z_1}\oplus\kappa_{1Z_1} (\mathbf{X}|_{S_{Z_1}}),\notag\\
        & \qquad \qquad \qquad \qquad \qquad \qquad \qquad \qquad \qquad \qquad \qquad M_{1\overline{Z}_1}\oplus\kappa_{1\overline{Z}_1}(\mathbf{X}|_{S_{\overline{Z}_1}})\notag\\
        & \, = I (  M_{1\overline{Z}_1};Z_1, Y_1^n, Y_2^n,  S_{Z_1},S_{\overline{Z}_1}, \kappa_{1Z_1}, \kappa_{1\overline{Z}_1}, M_{1Z_1}\oplus\kappa_{1Z_1} (\mathbf{X}|_{S_{Z_1}}))\notag\\
        & \quad + I (  M_{1\overline{Z}_1};M_{1\overline{Z}_1}\oplus\kappa_{1\overline{Z}_1}(\mathbf{X}|_{S_{\overline{Z}_1}})|Z_1, Y_1^n, Y_2^n,  S_{Z_1},S_{\overline{Z}_1}, \kappa_{1Z_1}, \kappa_{1\overline{Z}_1},\notag\\
        & \qquad \qquad \qquad \qquad \qquad \qquad \qquad \qquad \qquad \qquad \qquad M_{1Z_1}\oplus\kappa_{1Z_1} (\mathbf{X}|_{S_{Z_1}}))\notag\\
        & \overset{(b)}{=}I (  M_{1\overline{Z}_1};M_{1\overline{Z}_1}\oplus\kappa_{1\overline{Z}_1}(\mathbf{X}|_{S_{\overline{Z}_1}})|Z_1, Y_1^n, Y_2^n,  S_{Z_1},S_{\overline{Z}_1}, \kappa_{1Z_1}, \kappa_{1\overline{Z}_1},\notag\\
        & \qquad \qquad \qquad \qquad \qquad \qquad \qquad \qquad \qquad \qquad \qquad M_{1Z_1}\oplus\kappa_{1Z_1} (\mathbf{X}|_{S_{Z_1}}))\notag\\
        &\, = H (  M_{1\overline{Z}_1}\oplus\kappa_{1\overline{Z}_1}(\mathbf{X}|_{S_{\overline{Z}_1}})|Z_1, Y_1^n, Y_2^n,  S_{Z_1},S_{\overline{Z}_1}, \kappa_{1Z_1}, \kappa_{1\overline{Z}_1}, \notag\\
        & \qquad \qquad \qquad \qquad \qquad \qquad \qquad \qquad \qquad \qquad \qquad M_{1Z_1}\oplus\kappa_{1Z_1} (\mathbf{X}|_{S_{Z_1}}))\notag\\
        & \quad - H (  M_{1\overline{Z}_1}\oplus\kappa_{1\overline{Z}_1}(\mathbf{X}|_{S_{\overline{Z}_1}})|M_{1\overline{Z}_1},Z_1, Y_1^n, Y_2^n,  S_{Z_1},S_{\overline{Z}_1}, \kappa_{1Z_1}, \kappa_{1\overline{Z}_1}, \notag\\
        & \qquad \qquad \qquad \qquad \qquad \qquad \qquad \qquad \qquad \qquad \qquad M_{1Z_1}\oplus\kappa_{1Z_1} (\mathbf{X}|_{S_{Z_1}}))\notag\\
        &\, = H (  M_{1\overline{Z}_1}\oplus\kappa_{1\overline{Z}_1}(\mathbf{X}|_{S_{\overline{Z}_1}})|Z_1, Y_1^n, Y_2^n,  S_{Z_1},S_{\overline{Z}_1}, \kappa_{1Z_1}, \kappa_{1\overline{Z}_1},\notag\\
        & \qquad \qquad \qquad \qquad \qquad \qquad \qquad \qquad \qquad \qquad \qquad M_{1Z_1}\oplus\kappa_{1Z_1} (\mathbf{X}|_{S_{Z_1}}))\notag\\
        & \quad - H (\kappa_{1\overline{Z}_1}(\mathbf{X}|_{S_{\overline{Z}_1}})|M_{1\overline{Z}_1},Z_1, Y_1^n, Y_2^n,  S_{Z_1},S_{\overline{Z}_1}, \kappa_{1Z_1}, \kappa_{1\overline{Z}_1},\notag\\
        & \qquad \qquad \qquad \qquad \qquad \qquad \qquad \qquad \qquad \qquad \qquad M_{1Z_1}\oplus\kappa_{1Z_1} (\mathbf{X}|_{S_{Z_1}}))\notag\\
        & \overset{(c)}{\leq} n (r_1-\lambda')\notag\\
        & \quad -  H (\kappa_{1\overline{Z}_1}(\mathbf{X}|_{S_{\overline{Z}_1}})|M_{1\overline{Z}_1},Z_1, Y_1^n, Y_2^n,  S_{Z_1},S_{\overline{Z}_1}, \kappa_{1Z_1}, \kappa_{1\overline{Z}_1}, \notag\\
        & \qquad \qquad \qquad \qquad \qquad \qquad \qquad \qquad \qquad \qquad \qquad M_{1Z_1}\oplus\kappa_{1Z_1} (\mathbf{X}|_{S_{Z_1}}))\notag\\
        & \overset{(d)}{=} n (r_1-\lambda') -  H (\kappa_{1\overline{Z}_1}(\mathbf{X}|_{S_{\overline{Z}_1}})|\kappa_{1\overline{Z}_1}, Y_1^n|_{S_{\overline{Z}_1}})\notag\\
        & \overset{(e)}{\leq} n (r_1-\lambda') - n(r_1-\lambda') + \frac{2^{-n\lambda'}}{\ln 2}\notag\\
        & \, = \frac{2^{-n\lambda'}}{\ln 2},
    \end{align}
    
    where $(a)$ is due to $\mathbf{C}_1 = (S_0,S_1, \kappa_{10}, \kappa_{11}, M_{10}\oplus\kappa_{10} (\mathbf{X}|_{S_0}), M_{11}\oplus\kappa_{11}(\mathbf{X}|_{S_1}))$, $(b)$ follows from the fact that $M_{1\overline{Z}_1} \perp Z_1, Y_1^n, Y_2^n,  S_{Z_1},S_{\overline{Z}_1}, \kappa_{1Z_1}, \kappa_{1\overline{Z}_1}, M_{1Z_1}\oplus\kappa_{1Z_1} (\mathbf{X}|_{S_{Z_1}})$, $(c)$ follows since $H (  M_{1\overline{Z}_1}\oplus\kappa_{1\overline{Z}_1}(\mathbf{X}|_{S_{\overline{Z}_1}}))$ is $n(r_1-\lambda')$ bits long, $(d)$ follows from the independence of $\kappa_{1\overline{Z}_1}(\mathbf{X}|_{S_{\overline{Z}_1}})$ and\linebreak $M_{1\overline{Z}_1},Z_1, Y_1^n, Y_2^n,  S_{Z_1},S_{\overline{Z}_1}, \kappa_{1Z_1}, M_{1Z_1}\oplus\kappa_{1Z_1} (\mathbf{X}|_{S_{Z_1}})$ given $\kappa_{1\overline{Z}_1},Y_1^n|_{S_{\overline{Z}_1}}$,\linebreak and in $(e)$ we directly applied Lemma \ref{entropy hash}, so that $H_2(\mathbf{X}|_{S_{\overline{Z}_1}}|Y_1^n|_{S_{\overline{Z}_1}}=y_1^n|_{s_{\overline{z}_1}}) = \Delta(y_1^n|_{s_{\overline{z}_1}}) \geq n r_1$ since the size of the set $S_{\overline{Z}_1}$ is at least $n r_1$. Also, $l$ in Lemma \ref{entropy hash} is: $l = H (  M_{1\overline{Z}_1}\oplus\kappa_{1\overline{Z}_1}(\mathbf{X}|_{S_{\overline{Z}_1}})) = n(r_1-\lambda')$. 
    
    \parbreak Then, Eq. \eqref{final2: BC-nColl} tends to: 
    \begin{align*}
        \lim_{n\rightarrow\infty} I (  M_{1\overline{Z}_1};Z_1,Y_1^n, Y_2^n, \mathbf{C}_1)& \leq \lim_{n\rightarrow\infty} \frac{2^{-n\lambda'}}{\ln 2} = 0.
    \end{align*}
    The second term of \eqref{eq: mid} can be proved in the same fashion knowing that, $\mathbf{C}_2 = (S'\triangleq (S'_0,S'_1), \kappa_{20}, \kappa_{21}, M_{20}\oplus\kappa_{20} ((\mathbf{X}|_{S'})|_{S'_0}), M_{21}\oplus\kappa_{21}((\mathbf{X}|_{S'})|_{S'_1}))$. Then, \eqref{eq: mid} tends to zero as $n\to\infty$. 
    \item \begin{align}\label{eq: ind}
    I (Z_i; U, V_{\overline{i}}) & = I(Z_i;M_{i0},M_{i1},M_{\overline{i}0},M_{\overline{i}1}, Z_{\overline{i}}, X^n, Y_{\overline{i}}^n, \mathbf{C})\notag\\
    & = I(Z_i;M_{i0},M_{i1},M_{\overline{i}0},M_{\overline{i}1}, Z_{\overline{i}}, X^n, Y_{\overline{i}}^n, \mathbf{C}_i,\mathbf{C}_{\overline{i}})\notag\\
    & = I(Z_i;M_{i0},M_{i1},X^n, Y_{\overline{i}}^n, \mathbf{C}_i)\notag\\
    & \quad + I(Z_i;Z_{\overline{i}},M_{\overline{i}0},M_{\overline{i}1},\mathbf{C}_{\overline{i}}|M_{i0},M_{i1},X^n,Y_{\overline{i}},\mathbf{C}_i)\notag\\
    & = I(Z_i;M_{i0},M_{i1},X^n, Y_{\overline{i}}^n, \mathbf{C}_i).
    \end{align}
    Without loss of generality, consider the case of $i=1$:
\begin{align}\label{eq: perfect}
        I(Z_1;U,V_2)& = I(Z_1;M_{10},M_{11},X^n, Y_{2}^n, S_0,S_1, \kappa_{10}, \kappa_{11},M_{10}\oplus\kappa_{10} (\mathbf{X}|_{S_0}),\notag\\
        & \qquad\qquad\qquad\qquad\qquad\qquad\qquad\qquad\qquad M_{11}\oplus\kappa_{11}(\mathbf{X}|_{S_1}))\notag\\
        & = I(Z_1;M_{10},M_{11},X^n, Y_{2}^n, S_0,S_1, \kappa_{10}, \kappa_{11}, \kappa_{10} (\mathbf{X}|_{S_0}), \kappa_{11}(\mathbf{X}|_{S_1}))\notag\\
        & = I(Z_1;X^n, Y_2^n, S_0,S_1, \kappa_{10}, \kappa_{11}, \kappa_{10} (\mathbf{X}|_{S_0}), \kappa_{11}(\mathbf{X}|_{S_1}))\notag\\
        & \quad + I(Z_1;M_{10},M_{11}|X^n, Y_2^n, S_0,S_1, \kappa_{10}, \kappa_{11}, \kappa_{10} (\mathbf{X}|_{S_0}), \kappa_{11}(\mathbf{X}|_{S_1})) \notag\\
        & = I(Z_1;X^n, Y_2^n, S_0,S_1, \kappa_{10}, \kappa_{11}, \kappa_{10} (\mathbf{X}|_{S_0}), \kappa_{11}(\mathbf{X}|_{S_1}))\notag\\
        & = I(Z_1;X^n, Y_2^n, S_0,S_1)\notag\\
        & \quad + I(Z_1; \kappa_{10}, \kappa_{11}, \kappa_{10} (\mathbf{X}|_{S_0}), \kappa_{11}(\mathbf{X}|_{S_1})|X^n, Y_2^n, S_0,S_1)\notag\\
        & = I(Z_1;X^n, Y_2^n, S_0,S_1)\notag\\
        & = I(Z_1; S_0,S_1)\notag\\
        & = 0,
    \end{align}
    where the last equality holds since both sub-BECs act independently on each input bit and $|S_{0}|=|S_{1}|$. Similarly, for the case of $i=2$, we have:
    \begin{align*}
        I(Z_2;U,V_1)&=  I(Z_2;M_{20},M_{21}, X^n, Y_1^n, S'_0,S'_1, \kappa_{20},\kappa_{21},M_{20}\oplus\kappa_{20} ((\mathbf{X}|_{S'})|_{S'_0}),\notag\\
        & \qquad\qquad\qquad\qquad\qquad\qquad\qquad\qquad\qquad M_{21}\oplus\kappa_{21}((\mathbf{X}|_{S'})|_{S'_1}))\\
        & = I(Z_2;S'_0,S'_1)\\
        & = 0,
    \end{align*}
    where it follows from the same steps as the case of $i=1$. The perfect secrecy demonstrated above conforms entirely to the statements in Remark \ref{rem} and Corollary \ref{cor}.
    
    \item \begin{align*}
        I(Z_1, Z_2 ; U) &= I(Z_1;U) + I(Z_2;U|Z_1)\\
        & = I(Z_1;U) + I(Z_2;U,Z_1)\\
        & = I(Z_1;U) + I(Z_2;U)\\
        & = I(Z_1 ; M_{10}, M_{11}, M_{20}, M_{21}, X^n, \mathbf{C}) \\
        & \quad + I(Z_2 ; M_{10}, M_{11}, M_{20}, M_{21}, X^n, \mathbf{C})\\
        & \overset{(a)}{=} I(Z_1 ; M_{10}, M_{11}, X^n, \mathbf{C}_1) \\
        & \quad + I(Z_2 ; M_{20}, M_{21}, X^n, \mathbf{C}_2)\\
        & = I(Z_1 ; M_{10}, M_{11}, X^n, S_0, S_1, S', \kappa_0, \kappa_1, M_{10}\oplus\kappa_0(\mathbf{X}|_{S_0}),\notag\\
        & \qquad\qquad\qquad\qquad\qquad\qquad\qquad\qquad\qquad M_{11}\oplus\kappa_1(\mathbf{X}|_{S_1})) \\
        & \quad + I(Z_2 ; M_{20}, M_{21}, X^n, \mathbf{C}_2)\\
        &= I(Z_1 ; M_{10}, M_{11}, X^n, S_0, S_1, S', \kappa_0, \kappa_1, \kappa_0(\mathbf{X}|_{S_0}),\kappa_1(\mathbf{X}|_{S_1}))\\
        & \quad + I(Z_2 ; M_{20}, M_{21}, X^n, \mathbf{C}_2)\\
        & = I(Z_1 ; M_{10}, M_{11}, X^n, S_0, S_1, S', \kappa_0, \kappa_1) \\
        & \quad + I(Z_2 ; M_{20}, M_{21}, X^n, \mathbf{C}_2)\\
        & \overset{(b)}{=} I(Z_1 ; X^n, S_0, S_1, S') \\
        & \quad + I(Z_2 ; M_{20}, M_{21}, X^n, \mathbf{C}_2)\\
        &\overset{(c)}{=} I(Z_1 ; S_0, S_1, S') \\
        & \quad + I(Z_2 ; M_{20}, M_{21}, X^n, \mathbf{C}_2)\\
        &\overset{(d)}{=} 0,
    \end{align*}
    where $(a)$ is due to the fact that $Z_i$ is independent of $M_{\overline{i}0}, M_{\overline{i}1}, \mathbf{C}_{\overline{i}}$, $(b)$ follows since $M_{10}, M_{11}, \kappa_0, \kappa_1 \perp (Z_1, X^n, S_0, S_1, S')$, $(c)$ follows since $X^n \perp (Z_1, S_0, S_1, S')$, and $(d)$ follows since the channel acts independently on each input bit and $|S_0| = |S_1|$ for the first term, and the second term can be similarly proved to be zero.
    \item \begin{align}\label{eq: dep2}
        I &(M_{i0},M_{i1},M_{\overline{i}\overline{Z}_{\overline{i}}}, Z_i; V_{\overline{i}})_{i\in\{1,2\}}\notag\\
        &\,= I (M_{i0},M_{i1}, Z_i; V_{\overline{i}}) + I (M_{\overline{i}\overline{Z}_{\overline{i}}};V_{\overline{i}}|M_{i0},M_{i1}, Z_i)\notag\\
        &\,= I (M_{i0},M_{i1}, Z_i; Z_{\overline{i}},Y^n_{\overline{i}},\mathbf{C}) + I (M_{\overline{i}\overline{Z}_{\overline{i}}};V_{\overline{i}},M_{i1},M_{i2}, Z_i)\notag\\
        &\,= I (M_{i0},M_{i1}, Z_i; Z_{\overline{i}},Y^n_{\overline{i}},\mathbf{C}_i) + I (M_{i0},M_{i1}, Z_i;\mathbf{C}_{\overline{i}}| Z_{\overline{i}},Y^n_{\overline{i}},\mathbf{C}_i)\notag\\
        & \qquad\qquad\qquad\qquad\qquad\qquad\quad+ I (M_{\overline{i}\overline{Z}_{\overline{i}}};V_{\overline{i}},M_{i0},M_{i1}, Z_i )\notag\\
        &\overset{(a)}{=} I (M_{i0},M_{i1}, Z_i; Z_{\overline{i}},Y^n_{\overline{i}},\mathbf{C}_i) + I (M_{\overline{i}\overline{Z}_{\overline{i}}};V_{\overline{i}},M_{i0},M_{i1}, Z_i )\notag\\
        &\overset{(b)}{=} I (M_{i0},M_{i1}, Z_i; Z_{\overline{i}},Y^n_{\overline{i}},\mathbf{C}_i) + I (M_{\overline{i}\overline{Z}_{\overline{i}}};V_{\overline{i}}, Z_i)\notag\\
        &\overset{(c)}{\leq} I (M_{i0},M_{i1}, Z_i; Z_{\overline{i}},Y^n_{\overline{i}},\mathbf{C}_i) + I (M_{\overline{i}\overline{Z}_{\overline{i}}};V_i,V_{\overline{i}})\notag\\
        & \overset{(d)}{=} I (M_{i0},M_{i1}, Z_i; Z_{\overline{i}},Y^n_{\overline{i}},\mathbf{C}_i)\notag\\
        &\, = I(Z_i;Z_{\overline{i}}, Y_{\overline{i}}^n,\mathbf{C}_i) + I (M_{i0},M_{i1}; Z_{\overline{i}},Y^n_{\overline{i}},\mathbf{C}_i|Z_i)\notag\\
        & \overset{(e)}{=} I  (M_{i0},M_{i1}; Z_{\overline{i}},Y^n_{\overline{i}},\mathbf{C}_i|Z_i)\notag\\
        &\overset{(f)}{=} I  (M_{10},M_{11}; Z_{2},Y^n_{2},\mathbf{C}_1|Z_1)\notag\\
        & \,=  I  (M_{10},M_{11}; Z_{2},Y^n_{2},S_0,S_1, \kappa_{10}, \kappa_{11}, M_{10}\oplus\kappa_{10} (\mathbf{X}|_{S_0}),\notag\\
        & \qquad\qquad\qquad\qquad\qquad\qquad\qquad\qquad\qquad M_{11}\oplus\kappa_{11}(\mathbf{X}|_{S_1}))|Z_1)\notag\\
        & \,=  I  (M_{10},M_{11};Y^n_{2},S_0,S_1, \kappa_{10}, \kappa_{11}, M_{10}\oplus\kappa_{10} (\mathbf{X}|_{S_0}), M_{11}\oplus\kappa_{11}(\mathbf{X}|_{S_1}))|Z_1)\notag\\
        & \,=  I  (M_{10},M_{11}; M_{10}\oplus\kappa_{10} (\mathbf{X}|_{S_0}), M_{11}\oplus\kappa_{11}(\mathbf{X}|_{S_1}))|Z_1, Y^n_{2},S_0,S_1, \kappa_{10}, \kappa_{11})\notag\\
        & \,=  H  ( M_{10}\oplus\kappa_{10} (\mathbf{X}|_{S_0}), M_{11}\oplus\kappa_{11}(\mathbf{X}|_{S_1}))|Z_1, Y^n_{2},S_0,S_1, \kappa_{10}, \kappa_{11})\notag\\
        & \quad - H  ( \kappa_{10} (\mathbf{X}|_{S_0}), \kappa_{11}(\mathbf{X}|_{S_1}))|Z_1, M_{10}, M_{11}, Y^n_{2},S_0,S_1, \kappa_{10}, \kappa_{11})\notag\\
        & \overset{(g)}{\leq} 2n(r_1-\lambda')\notag\\
        &\quad - H  ( \kappa_{10} (\mathbf{X}|_{S_0}), \kappa_{11}(\mathbf{X}|_{S_1}))|Z_1, M_{10}, M_{11}, Y^n_{2},S_0,S_1, \kappa_{10}, \kappa_{11})\notag\\
        & \overset{(h)}{\leq}  \frac{2^{-n\lambda'}}{\ln 2},
    \end{align}
    where $(a)$ and $(b)$ follow from the Markov chains $M_{i0},M_{i1},Z_i-Z_{\overline{i}},Y_{\overline{i}}^n,\mathbf{C}_i-\mathbf{C}_{\overline{i}}$, and $M_{\overline{i}\overline{Z}_{\overline{i}}}-V_{\overline{i}},Z_i-M_{i0},M_{i1}$, respectively, $(c)$ follows since the second term asymptotically tends to zero, followed immediately by \eqref{goals: BC-Coll-2}, $(d)$ follows from \eqref{goals: BC-Coll-3}, $(e)$ is due to \eqref{goals: BC-Coll-4}, in $(f)$ we set $i=1$ without loss of generality, $(g)$ is due to this fact that $\kappa_{10} (\mathbf{X}|_{S_0})$ and $ \kappa_{11}(\mathbf{X}|_{S_1})$ are $n(r_i-\lambda')$ bits long, and $(h)$ follows from the same argument as \eqref{final2: BC-nColl}.

    \parbreak
    Now it is straightforward to show that:
    \[
    \lim_{n\to\infty}   I (M_{i0},M_{i1},M_{\overline{i}\overline{Z}_{\overline{i}}}, Z_i; V_{\overline{i}})_{i\in\{1,2\}} = \lim_{n\to\infty} \frac{2^{-n\lambda'}}{\ln 2} = 0.
    \]
    \end{itemize}
    
\section{Proof of Proposition \ref{prop2}\label{app: thm: lower: coll}}
Since Protocol 2 is a successive protocol comprising two sub-OT protocols (Protocol-I, Protocol-II), using a parameter $\gamma \in [0,1]$, 
one can construct a new protocol by \emph{time-sharing}. Specifically, Protocol-I is 
executed over $\gamma n$ channel uses, followed by Protocol-II over the remaining 
$(1-\gamma)n$ channel uses. If Protocol-I and Protocol-II achieve rates $R_1$ and $R_2$, 
respectively, then the resulting protocol achieves the convex combination $R = \gamma R_1 + (1-\gamma) R_2$. 

\parbreak We must bound two conditional smooth min entropy quantities:
\begin{equation}\label{eq: first'}
    H^\varepsilon_{\infty}(\mathbf{X}|_{S_{1\overline{Z}_1}}|h_{10}(R_{10}, \mathbf{X}|_{S_{10}}), h_{11}(R_{11}, \mathbf{X}|_{S_{11}}), Y_1^n, R^{(1)}, T^{(1)}),
\end{equation}
\begin{equation}\label{eq: second'}
    H^\varepsilon_{\infty}(\mathbf{X}|_{S_{2\overline{Z}_2}}|h_{20}(R_{20}, \mathbf{X}|_{S_{20}}), h_{21}(R_{21}, \mathbf{X}|_{S_{21}}), Y_2^n, R^{(2)}, T^{(2)}),
\end{equation}
Based on \cite{Winter1}, we have (Protocol-I $\to$ Protocol-II): 
\begin{align*}
    H^\varepsilon_{\infty}(\mathbf{X}|_{S_{1\overline{Z}_1}}|h_{10}(R_{10}, \mathbf{X}|_{S_{10}}), h_{11}(R_{11}, \mathbf{X}|_{S_{11}}), Y_1^n, R^{(1)}, T^{(1)}) &\geq (1-p_1)nH(X)\notag\\
    & \quad- s_1 n - \delta n,\\
    H^\varepsilon_{\infty}(\mathbf{X}|_{S_{2\overline{Z}_2}}|h_{20}(R_{20}, \mathbf{X}|_{S_{20}}), h_{21}(R_{21}, \mathbf{X}|_{S_{21}}), Y_2^n, R^{(2)}, T^{(2)}) &\geq (1-p_2)nH(X)\notag\\
    & \quad - s_2 n - \delta n,
\end{align*}
for any $\delta \geq (\alpha+\alpha' + 2\eta +4\sqrt{\alpha})> 0$. Due to the Chernoff bound, we know that the probability that Bob aborts the protocol in step $(2)$ tends to zero as $n\rightarrow \infty$. Protocol 2 fails in step $(5)$, if there is more than one, such as $(\hat{\mathbf{x}}|_{S_{Z_i}})$ where $h_i(\mathbf{X}|_{S_{Z_i}}) = h_i(\hat{\mathbf{X}}|_{S_{Z_i}})$. We know that if all players are honest, then $(\mathbf{Q}_1|_{S_{Z_1}},\mathbf{Q}_2|_{S'_{Z_2}}) =( \mathbf{Y}_1|_{S_{Z_1}},\mathbf{Y}_2|_{S'_{Z_2}})$ with probability exponentially close to one and the number of sequences $\hat{\mathbf{x}}|_{S_{1Z_1},S_{2Z_2}}$ jointly typical with $(\mathbf{q}_1|_{S_{Z_1}},\mathbf{q}_2|_{S'_{Z_2}})$ can be bounded from above so that if one of the sequences $\hat{\mathbf{x}}|_{S_{Z_1}}$ is not typical with $(\mathbf{y}_1|_{S_{Z_{1}}}, \mathbf{y}_2|_{S'_{Z_2} })$: $2^{\lvert S_{Z_1} \rvert(H(X,Y_1,Y_2)- H({Y_1,Y_2})+ \delta'}) =  2^{n(1-p_1)(H(X|Y_1,Y_2)+\delta')}$, $\delta'>0$. From \eqref{eq: hashs}, we know that $p\leq 2^{-s_1n}2^{(1-p_1)n(H(X|Y_1,Y_2)+\delta')}$, then $s_1 > (1-p_1) H(X|Y_1,Y_2)$. Similarly, for Protocol-II, we have: $s_2 > (1-p_2)(2p_1-1) H(X|Y_1,Y_2)$ since the unused segment
of Alice’s transmissions is completely erased for Bob-1 and is $\lvert S'\rvert\approx n-2(1-p_1)n = (2p_1-1)n$ bits long. Finally, by \eqref{Lower}, proved by Ahlswede and Csiszár, and  for the order of Protocol-I $\to$ Protocol-II, the rate region is:
\begin{align*}
\mathcal{R}_{\text{I$\to$II}} = \left\{\ 
\begin{aligned}
R_1 &\leq p_2.\min\{p_1,1 - p_1\} \\
R_2 &\leq (2p_1-1) .\min\{p_2,1 - p_2\}
\end{aligned}
\right\},
\end{align*}
and for the opposite order: 
\begin{align*}
\mathcal{R}_{\text{II$\to$I}} = \left\{\ 
\begin{aligned}
R_1 &\leq (2p_2-1).\min\{ p_1,1 - p_1\} \\
R_2 &\leq p_1.\min \{p_2, 1 - p_2\}
\end{aligned}
\right\}.
\end{align*}
Finally, the rate region for Protocol 2 can be obtained as follows:
single-user OT for Bob-1 can achieve rate $R_1\leq p_2.\min\{p_1,1-p_1\}$ and consumes exactly a $\min\{p_1,1-p_1\}$-fraction of indices that are erased at Bob-1 (as masks). Symmetrically, a single-user OT for Bob-2 can achieve $R_2\leq p_1.\min\{p_2,1-p_2\}$ and consumes a $\min\{p_2,1-p_2\}$-fraction of indices that are erased at Bob-2. Consider the two-phase time-sharing Protocol 2; the order matters because the second phase can only use the leftover erasures of the other receiver.
\begin{itemize}
    \item Protocol-I $\to$ Protocol-II: Phase 1 gives $R_1 = p_2.\min\{p_1,1-p_1\}$ and uses $1-p_1$ of Bob-1’s erased indices. Leftover Bob-erased fraction is $p_1-(1-p_1) = (2p_1-1)$, so Bob-2’s Phase 2 rate is at most $(2p_1-1).\min\{p_2,1-p_2\}$. Thus 
    \[
    R_1 + R_2 \leq p_2.\min\{p_1,1-p_1\} + (2p_1-1).\min\{p_2,1-p_2\}. 
    \]
    \item  Protocol-II $\to$ Protocol-I: Phase 1 gives $R_2 = p_1.\min\{p_2,1-p_2\}$ and uses $\min\{p_2,1-p_2\}$ of Bob-erased indices. Leftover Bob-2 erased fraction is $p_2-(1-p_2) = (2p_2-1)$, so Bob-1’s Phase 2 rate is at most $(2p_2-1).\min\{p_1,1-p_1\}$. Thus 
    \[
    R_1 + R_2 \leq p_1.\min\{p_2,1-p_2\} + (2p_2-1).\min\{p_1,1-p_1\}.
    \]
\end{itemize}
Finally, we get: 
\[
R_1 + R_2 \leq p_1.\min\{p_2,1-p_2\} + p_2.\min\{p_1,1-p_1\} - \min\{p_1,1-p_1\}.\min\{p_2,1-p_2\}.
\]
This completes the proof. 

\bibliographystyle{IEEEtran}
\bibliography{references}
\end{document}